\documentclass[11pt]{article}
\pdfoutput=1

\usepackage[dvipsnames]{xcolor}
\usepackage{hyperref}  
\hypersetup{
     colorlinks,
     citecolor=OliveGreen,
     filecolor=black,
     linkcolor=blue,
     urlcolor=black,
     linktoc=all
}
\usepackage[shortlabels]{enumitem} 
\usepackage{amsmath, amssymb, amsthm, mathtools}

\newcommand{\ignore}[1]{}
\newtheorem{theorem}{Theorem}
\newtheorem{lemma}{Lemma}[section]

\newtheorem{claim}{Claim}
\newtheorem{observation}{Observation}

\newtheorem{definition}{Definition}
\newtheorem{remark}{Remark}

\usepackage{bbm}
\usepackage{mathrsfs}
\usepackage{latexsym}
\usepackage{microtype}
\usepackage{xspace}
\usepackage{framed}
\usepackage{tocloft}
\PassOptionsToPackage{dvipsnames,usenames,table}{xcolor}
\usepackage{tikz}
\usetikzlibrary{calc,math}
\usetikzlibrary{arrows,shapes}
\usepackage{lineno} 
\newbox\mytempbox
\tikzstyle{embeds} = [->, >=open triangle 45]

\usepackage{xpunctuate}
\usepackage[all,british]{foreign} 
\redefnotforeign[ie]{i.e\xperiodafter} 
\redefnotforeign[eg]{e.g\xperiodafter}

\newenvironment{tightcenter}
 {\parskip=0pt\par\nopagebreak\centering}
 {\par\noindent\ignorespacesafterend}

\usepackage{marginnote} 

\usepackage[draft,english,silent]{fixme} 
\fxsetup{theme=color,layout=marginnote,innerlayout=noinline}

\usepackage{amsmath, amsthm, amssymb, amsfonts} 
\usepackage{thmtools} 
\usepackage{mathtools}

\usepackage[full,small]{complexity}

\usepackage{environ}

\usepackage{ifthen}
\newboolean{NotfullVersion}

\setboolean{NotfullVersion}{false} 
\newcommand{\Z}{\mathbb Z}

\newtheorem*{theorem*}{Theorem}
\newlength{\RoundedBoxWidth}
\newsavebox{\GrayRoundedBox}
   {\setlength{\RoundedBoxWidth}{\textwidth-4.5ex}
    \def\boxheading{#1}
    \begin{lrbox}{\GrayRoundedBox}
       \begin{minipage}{\RoundedBoxWidth}%
   }{%
       \end{minipage}
    \end{lrbox}%
    \begin{tightcenter}%
    \begin{tikzpicture}%
       \node(Text)[draw=black!20,fill=white,rounded corners,%
             inner sep=2ex,text width=\RoundedBoxWidth]%
             {\usebox{\GrayRoundedBox}};
        \coordinate(x) at (current bounding box.north west);
        \node [draw=white,rectangle,inner sep=3pt,anchor=north west,fill=white] 
        at ($(x)+(6pt,.75em)$) {\boxheading};
    \end{tikzpicture}
    \end{tightcenter}\vspace{0pt}%
    \ignorespacesafterend
}

\usepackage[longend,vlined,plain,linesnumbered]{algorithm2e}

\SetKwInput{KwInput}{Input}
\SetKwInput{KwOutput}{Output}
\SetKwFor{RepTimes}{repeat}{times}{end}

\title{Deterministic Online Embedding of Metric Spaces into
  Low Dimensional Spaces}

\author{Noam Licht\thanks{ Dept. of CS, University of Haifa,  Israel }\and
 Ilan Newman\thanks{University of Haifa, Israel. Email:
    \href{mailto:ilan@cs.haifa.ac.il}{ilan@cs.haifa.ac.il}.}\and
Yuri Rabinovich\thanks{Dept. of CS, University of Haifa,  Israel,
  Email: yuri@cs.haifa.ac.il}
}

\date{\today}

\begin{document}

\maketitle

\begin{abstract}
We study online embeddings of metric spaces into Euclidean spaces of a constant dimension $d>1$, against an adaptive adversary. While the case of $d=1$ is well understood, for higher dimensions  little is known.
In particular, even  for $d=2$ it remains unknown  whether the worst-case
distortion  grows exponentially with the number of exposed points, as it does 
in the case for the line, or whether it is polynomial, as in the case for unbounded $d$.

Our first result is about fixed {\em solid} graphs,
i.e., $K_5$, whose edges are solid intervals, equipped with the shortest-path metric. 
We show that if the input points arrive from such a metric space, they can indeed be online-embedded into ${\mathbb R}^2$ with a polynomial distortion. This refutes the previously believed conjecture that the topological non-embeddability of $K_5$ into the plane
could be exploited for establishing exponential lower bounds.

The second results is about online embeddings of tree metrics of a
certain type, including, e.g.,  ultrametrics and  HST's. Somewhat
surprisingly, we show that for metrics from this class the worst-case
online embedding into ${\mathbb R}^d$ is not much worse that the offline
embedding, both being $n^{\Theta(1/d)}$, and this holds even when $d =
\Theta(\log n)$. This is in a stark contrast to the more common situation where the online-offline gap is typically huge, and even exponential. 
This result  allows us to transfer  results about 
probabilistic embeddings of metrics into HST's to low-dimensional
Euclidean spaces, 
in an almost optimal possible manner.
\end{abstract}



\definecolor{Blue}{rgb}{0.3,0.3,0.9}
\definecolor{XX}{rgb}{0.9,0.3,0.3}
\definecolor{Green}{rgb}{0.1,0.6,0.2}
\definecolor{Black}{rgb}{0.0,0.0,0.0}

\def \qed {\hspace*{0pt} \hfill {\quad \vrule height 1ex width 1ex depth 0pt}
 \medskip}

\def \nset {\{0, \ldots,n-1\}}
\def \Dn {d^{(n)}}
\def \D {diam}

\newcommand{\N}{\ensuremath{\mathbb N}}
\newcommand{\F}{\ensuremath{\mathcal F}}
\newcommand{\f}{\ensuremath{\mathbb F}}

\newtheorem{inv}{Invariant}[section]
 \newtheorem{notations}[theorem]{Notations}


\bibliographystyle{plain}


%
%
\section{Introduction}\label{sec:1}
The modern theory of low-distortion embeddings of finite metric
spaces began to take shape with the
appearance of classical results of Johnson and
Lindenstrauss~\cite{JL}\footnote{\em Any $n$-point Euclidean metric
  can be efficiently embedded into
  $\ell_2^{{{\log n} / {\epsilon^2}}}$ with
  $(1+\epsilon)$-distortion.}  and Bourgain~\cite{Bou}\footnote{ \em
  Any $n$-point metric can be efficiently embedded into Euclidean
  space with distortion $O(\log n)$.}, in the last decades of the
20'th century. It soon became clear that this theory provides powerful
tools for numerous theoretical and practical algorithmic
problems. Today, it has evolved into a well-established discipline employing advanced mathematical concepts, and playing a crucial role in solving complex algorithmic problems across a variety of domains.

The study of online metric embedding is a relatively recent development. In this setting, data points are presented one by one, and decisions must be made as they arrive. The goal, as before, is to preserve the distances or metric properties as much as possible. While the online setting is better suited for needs of modern real-time data processing, achieving good guarantees on metric distortion in this setting becomes significantly harder.

The first publication explicitly dedicated to online embeddings was~\cite{IMSZ} from 2010. One of the key observations of this paper was that a large part of Bartal's probabilistic offline embedding procedure into Hierarchically Well Separated Trees (henceforth, {\em HST}'s) of~\cite{Bar} can be
implemented online. This line of research, continued in~\cite{BarFan}, and strengthened in~\cite{BhFiTo},
yields, e.g., that:
{\em Any  metric space $(X,d)$ can be probabilistically online embedded into a distribution of a non-contracting ultrametrics (and hence tree-metrics) with dilation $O(\log q \cdot \log \min(q,\Delta))$. Here $q$ is
the number of exposed points, and it does not need to be known in advance.~ $\Delta$ stands for the {\em aspect ratio} of $(X,d)$, i.e., the ratio between the largest and the smallest distances therein. The result is essentially tight.}

Importantly, in this probabilistic setting one considers only the {\em expected} dilation of the distances. The guarantee on distortion holds with high probability but is not assured. More importantly,
the adversary is assumed to be non-adaptive, i.e., she must choose its input metric $(X,d)$ before the exposure/embedding dialogue starts. This is a severe restriction.

In this paper we are interested in the deterministic version of online embeddings, where the embedding is produced according to a deterministic protocol, and the adversary is allowed to be adaptive. That is, the 
metric $(X,d)$ does not need to be set in advance, and the choice of
the next exposed point and its distances may depend on the embedding
produced so far. This setting was studied explicitly in~\cite{NR1},
where tight exponential bound for online embedding of $(X,d)$ into the line was proved, and $(1+\epsilon)$-distortion  online embedding into
$\ell_\infty$ of a huge dimension. In addition, ~\cite{NR1} establishes a lower bound of $\Omega(n^{0.5})$ (implicit already in~\cite{NR}) for distortion of an online deterministic embedding into $\ell_2$ of unbounded dimension. Interestingly, the "hard" input metric is quite simple, and it offline embeds into the line with a constant distortion. An essentially matching upper bound was recently obtained in~\cite{BhFiTo} by means of an elegant novel construction.

The focus of the present paper are deterministic online embeddings into $\ell_2^k$ of a low, and in particular constant, dimension. 

Our first result deals with embeddings into the plane. Currently, it is not even known whether the deterministic
online embedding into the plane can ensure distortion\footnote{Distortion is the product of the worst metric
  expansion and the worst contraction. It is defined formally in
  Section \ref{sec:notation}.} polynomial in the number of the exposed points,
or  exponential distortion is unavoidable, as it is the case of embedding into the line. Let us remark that for probabilistic online embedding a polynomial (expected) distortion is indeed feasible, e.g., by randomly
projecting into the plane the embedding of~\cite{BhFiTo}. To put this
in context, the worse {\em offline} distortion of embedding a metric into the plane is linear in its size.
 
A standard technique for lower-bounding the distortion of 
embedding of $(X,d)$ into ${\mathbb R}^2$, both offline and online, is to exploit the topological
non-embeddability of $(X,d)$ into ${\mathbb R}^2$. In the offline setting,
one usually takes an $\epsilon$-net of $(X,d)$ and argues that it
cannot be embedded too well. E.g., such an argument was
used in \cite{SBDGIRRR} for (offline) embedding submetrics of the $2$-dim sphere into ${\mathbb R}^2$. 
In the online setting, one starts with a small $\epsilon$-net, allocates a small subset of $X$
whose embedding is problematic, and locally amplifies the error. This approach
was used, e.g., in~\cite{NR1} to establish exponential lower
bounds for deterministic online embedding into tree metrics.  

In view of the results in \cite{NR1}, it seemed  that the topological
non-embeddability into ${\mathbb R}^2$ of e.g., the {\em solid} $K_5$ could be
exploited, in a similar manner, to establish exponential lower bounds on
embedding its submetrics into the plane. Such metrics indeed embed badly offline. 
Our first finding is that, surprisingly, this topological obstacle can be handled at a
relative small cost to the distortion. 
The following basic problem captures the topological obstacle, and will serve as a gadget for further
constructions.

The {\em two-path game}  is a two players game  whose board is 
the unit square $[-\frac{1}{2},\frac{1}{2}]^2 \subset {\mathbb R}^2$. 
There are two unit-length paths: a blue path and a red one. 
They form a metric space, where the distance between two points of the same color is
their distance in the corresponding unit-length interval, while the distance between any 
blue and any red point is $1$. The endpoints of the blue path are glued
to $(-\frac{1}{2},0), (\frac{1}{2},0)$ respectively, and
the endpoints of the red path are glued to $(0, -\frac{1}{2}), (0,\frac{1}{2})$
respectively, see Figure \ref{fig:101}. Both path must lie in the board.
 The Challenger in his turn exposes points at his will, specifying their position on the paths.
The Embedder, in his turn, must place the exposed point on the board.
The Challenger's goal is to maximize the metric distortion, the Embedder  goal is to keep it low.

A standard analysis shows that if the $q$ points are given in advance, the Embedder can ensure 
distortion of $O(q)$, and this is tight. What about the online situation? 
Assuming that $q$ is known in advance, we show that Embedder can ensure $O(q^2)$ distortion.  
When $q$ is not known, a slightly weaker upper bound of $O(q^2\log^{2+\epsilon}q)$ can  be ensured 
at every round $q$.

\begin{figure}
\begin{center}
\includegraphics[scale = 1]{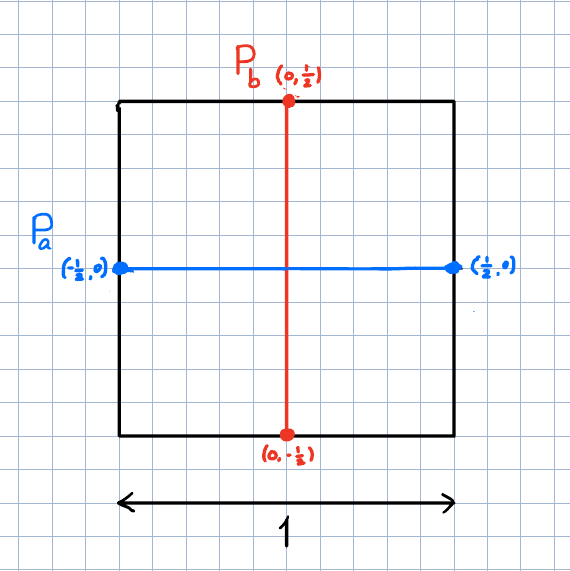}
\end{center}
\caption{The horizontal path $P_a$ and vertical path $P_b$ should be
  embedded in the grid with endpoints as shown.}
\label{fig:101}
\end{figure}

We proceed to analyze the case when the input points are drawn from a fixed solid graph $G$, i.e., the edges
are unit length intervals glued at the vertices, and the distances are induced by the shortest-path metric of $G$.
To avoid graph-automorphism issues, we assume that the vertices and edges are uniquely labeled, and that the
input points are specified by their position on the edges. Using the
ideas from the two-paths game, we construct 
a deterministic online embedding of the exposed points that is quadratic in $q$, and polynomial in $|V(G)|$.
We further study the weighted solid $G$'s, and obtain similar upper bounds.

Our second direction of study is embedding of ultrametrics and  a certain generalization of $HST$ 
into $\ell_2^k$. It is well known that $q$ points from an ultrametric embed
isometrically into $\ell_2^{q-1}$. Due to rigidity
of $\ell_2$, i.e., the uniqueness of the embedding up to translations, rotations and reflections, this
embedding can be produced online as well. The optimal offline embeddings of ultrametrics into $\ell_2^k$
for $k$'s logarithmic in $q$ are studied in~\cite{BCIS}, establishing an essentially tight distortion of $k^{O(1)}q^{1/k}$. For $k$'s sub-logarithmic in $q$, the minimal dimension required for $(1+\epsilon)$ embedding is studied
in~\cite{BarLinManNa}. In view of the powerful results mentioned above about the (probabilistic) embedding of metric spaces into a distribution of HST's, it make sense to ask also about deterministic online embeddings
of such metrics into $\ell_2^k$. 

For ultrametrics as an input, we show that $q$ online exposed points
can be embedded  into $\ell_2^{k}$ with distortion $O(q^{(2/k)})$. In
particular, ensuring a constant distortion for embedding into $\ell_2^{\theta(\log q)}$.  This is surprisingly close to
the optimal offline bounds. However, here we  assume that $q$ is known in advance. Removing this requirement,
and possibly improving the distortion, is left for the future study. In addition, we introduce a natural class of metrics called $\alpha$-trees that generalize the HST's, and obtain for them similar bounds.

To sum up, our contribution is mainly in showing that some
non-trivial metric classes can be effectively embedded online, even in low dimensional Euclidean space contrary to prior expectations. 
It is our hope that these findings will be of help in the future study of online embedding arbitrary metric spaces in 
$\ell_2^k$.

\section{Notations and Preliminaries}\label{sec:notation}

{\em Distortion of embedding:} Let  $(X,d), ~ (Y, \mu)$ be two metric spaces and $\phi: X \longrightarrow Y$ an
embedding of $(X,d)$ into $(Y,\mu)$.  The standard definition of (multiplicative) 
metric distortion of $\phi$ is given by
\vspace{-0.4cm}
\[
dist(\phi) =  \max_{x,y \in X}\frac{\mu(\phi(x),\phi(y))}{d(x,y)} \cdot
\max_{u,v \in X}\frac{d(u,v)}{\mu(\phi(u),\phi(v))}
\]
The first term of this product is called the expansion of $\phi$, and the second term is called the contraction. In the offline setting, the distortion of embedding $(X,d)$ into
$(Y,\mu)$ can also be defined as the minimal possible expansion of a non-contracting $\phi$, or, equivalently, as a minimal possible contraction of a non-expanding $\phi$. 
In the online setting, the above two definitions are both valid, but not necessarily equivalent, 
 as an appropriate scaling might
not be known in advance.

\vspace{0.2cm}
\noindent

{\bf Solid graph metrics:} Let  $G=(V,E)$ be a connected graph with nonnegative weights (lengths) $\ell(e)$ on the edges. It may contain parallel edges and self-loops. Its "solid" one-dimensional realization is obtained by treating each edge $e \in E$ as one-dimensional intervals of length $\ell(e)$,  whose endpoints are glued at the vertices as in $G$. The solid graph equipped with its geodesic (i.e., shortest path) metric, is an {\em infinite} metric space, whose points are labeled by the edge containing them, and their position on this edge. The distance between two such points can be easily computed by their labels, and the  edge weights $\ell(e)$ of $G$. 

In the online setting, embedding points from the solid $(X_G,\mu)$ induced by a given weighted graph
$G$, when a point $x$ lying on the edge $e=(u,v)$ is exposed, 
we first ensure that the corresponding endpoints $u$ and $v$ are already embedded, and only then address the embedding of $x$.
We also assume 
 that $G,l$ are known in advance.    We
do so to avoid dealing with graph isomorphism problems, although most results
for online embedding of $X_G$ will have a similar upper bounds  even if the "end points" of the edge in which the given point lies are not given.

\vspace{0.2cm}
\noindent
{\bf Ultrametrics:}
A metric space $(X,d)$ is ultrametric, abbreviated UM,  on a {\em finite} set $X$ if for
every three points $x,y,z \in X$, $~d(x,z) \leq \max \{d(x,y),
d(y,z)\}$. It is folklore (and easy) that each of the following
conditions  characterize UM.

(a) - $(X,d)$ is UM if there exists a weighted rooted tree $T=(V,E)$ with root $r$,
and leaves $X\subset V$, such that $d_T(r,x)$ is the same for every $x
\in X$, and $d  = d|_X$ is the metric induced  by $d_T$ on the leaves $X$.

(b) - $X$ is a subset of the vertex set of a weighted undirected connected graph
$G =(V,E)$, and such that $d(x,y) = \min_{P_{x,y}} \max_{e \in
  P_{x,y}} w(e),$ where $P_{xy}$ ranges over all paths from $x$ to $y$.

By (a), UM is a special case of a tree metric.  It is quite obvious
that any metric on a finite set $X$ of size $n$ can be embedded
(offline) into UM with distortion at most $n-1$. This is so as the
graph $G=(X,E)$ for which $w(x,y) = d(x,y)$ represent a UM $d'$ on $X$
(by (b) above). Since the weight of the heaviest edge in any path is
at most the weighted length of the path and at least this length
divided by $n-1$, we conclude that $d/ (n-1) \leq d' \leq d$.
Moreover, this is best possible as shown by the the unit
weighted path of length $n$.

\vspace{0.2cm}
\noindent
{\bf $\alpha$- Trees:} 
        Let $T=(V,E)$ be an edge-weighted  rooted tree with root $r
        \in V$, and $0 < \alpha  < 1$.  $T$ is an $\alpha$- tree
        ($\alpha$-tree) if for any root-to-leaf path $(r=v_0, v_1,
        \ldots ,v_\ell)$, $~ w(v_i,v_{i+1}) \leq \alpha \cdot
        w(v_{i-1},v_i)$.

        The metric induced by $\alpha$-trees is obviously a tree
        metric. It is  a strict generalization
          of Bartal's  HST \cite{Bar}, since we do not require that all edges going
          out from a node have the same weight. As such, they are interesting metric spaces, and they have played an important role in probabilistic offline embedding of metric spaces into Euclidean spaces.  $\alpha$-trees   are mostly interesting
          when $\alpha$ is bounded away from $1$, e.g.,
          $\alpha=1/2$. In this case, it is easy to see that the metric induced by an $\alpha$-tree  is close to an $HST$. Namely, can be embedded into $HST$ with distortion $\frac{\alpha}{1-\alpha}$.
    It worth also mentioning that
         an ultrametric is at most $\alpha^{-1}$ far from an $\alpha$-HST.
         This does not imply, however, that the online embedding of an $\alpha$-tree is, 
         up to constants, as easy as that of $\alpha$-HST's. We find this class of metrics  
         interesting in its own right, while its online embedding is considerably more
        complicated than that of ultrametrics.

\vspace{0.2cm}
\noindent
{\bf The 2-Paths game}
Let $P_a = [0,1]\times \{a\}$ and $P_b = [0,1]\times \{b\}$ equipped
with with the
natural line metric $d$ on each. We define  the
following metric space  $(P_a \cup P_b,\mu)$ to be:   for every $x,y$ on the same path  $\mu(x,y)= |x-y|$. For $x \in P_a, y \in P_b$ $\mu(x,y)=1$.

The goal of the game  is to {\em deterministically online} map $(P_a \cup P_b, \mu)$ into the unit square 
$S=[-\frac{1}{2},\frac{1}{2}]^2
\subseteq \mathbb{R}^2$  with small metric distortion with respect to  $\mu$. Namely,   $q$
adversarially labeled points in $[0,1] \times \{a,b\}$  are exposed, one at a
time. In turn, each point is
 embedded into $S$.
The crux is  that the endpoints of $P_a$ and $P_b$ are restricted to be mapped as follows: $(0,a)$ and $(1,a)$ to $(-\frac{1}{2},0)$ and $(\frac{1}{2},0)$, respectively, and $(b,0)$ and $(b,1)$ to $(0,-\frac{1}{2})$ and $(0,\frac{1}{2})$,
respectively, see Figure \ref{fig:101}.

\section{Our Results}
{\bf The 2-Paths game:} It is easy to see that  any $q$ points in $(P_a
\cup P_b,\mu)$ can be embedded  {\em offline} into the unit
planar box $S=[-1/2,1/2]^2$ with the restrictions imposed on the
endpoints and distortion  $O(q)$.
We show, 
\begin{theorem} \label{MainTheorem}
  There exists an {\bf online }procedure  that can embed any $q$  points in the
  2-path game,  with distortion $ O(q^2)$, when $q$ is specified in
  advance. If $q$ is not known, a distortion
  $O(n^2\log^{2+\epsilon} n)$ is achieved at each step $n=1, \ldots$,  for any  $\epsilon <
  1$.
\end{theorem}

\vspace{0.2cm}
\noindent
{\bf Solid graph metrics:} 
We  use  the results for the 2-path game in order to online
embed more complex structures into the plane with relatively low
distortion. The idea is that the 2-path online algorithm allows us to  "cross" lines. 
The metrics in this section are of the form
$(X_G,\mu)$, where $G$ is a labeled (weighted / unweighted) graph that
is known in advance. The embedding algorithms have the following 2-phase form:
\begin{enumerate}
    \item A prepossessing phase: this is an offline phase which gets as
    input $G$.
    
    This phase creates a high-level mapping, using empty positions for
    the "1-intersection" gadget. The image of every point of the
    metric that is not inside some "1-intersection" territory will be
    determined in advance in this phase.
    
    \item The online phase: for a next exposed point we check whether
      it is inside some "1-intersection" territory. If so we map the
      point according to the "1-intersection" algorithm. Otherwise its
      image  has already been determined in the previous phase.
    \end{enumerate}

As a simple example consider the metric that is induced by a solid
unweighted $K_5$. Since $K_5$ is not planar, it can be easily verified that any
{\em offline} embedding of uniformly spaced $q$ points inside its
edges has $\Omega(q)$ distortion. A straight forward
implementation of the strategy  above  results in an
{\em online} embedding of {\em any} $q$ points with distortion $O(q^2)$. 
The reader is referred to Figure \ref{fig:5} and the
explanation below it.

\begin{figure}
\begin{center}
\includegraphics[scale = 0.35]{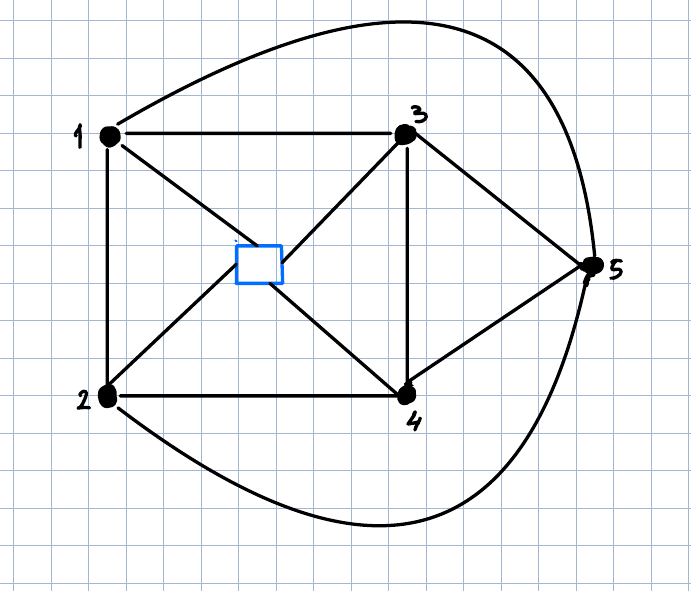}
\end{center}
\caption{The embedding of the metric in solid $K_5$. The blue
  rectangular zone is
  a 1-intersection territory in which the relevant points on the edge
  $(2,3)$ and the edge $(1,4)$ are mapped, using the 2-paths game.
Other points embed linearly along
  the edges that are pre-embedded in advance.}
\label{fig:5}
\end{figure}

Using the strategy outlined above (and a
 generalization of the 2-paths game to $k$ paths), we
obtain the following  result for {\em unit-weighted}
graphs.
Let $G$ be {\em unit} weighted graph that is known in advance, and
$(X_G,\mu)$ the corresponding solid metric space.

\begin{theorem}\label{thm:27}
    There exists an online algorithm that can embed $q$ points  from
 $X_G$ into ${\mathbb R}^2$ with distortion $O(mD \cdot 
q^2)$, where $G$ is unit-weighted with diameter $D$, and $|E(G)| =m$.  
\end{theorem}
We note that the 
diameter of the input metric is bounded by $|V(G)|$, but  the aspect ratio of $X_G$ is 
 unbounded. 
For a worst case of $D= \theta(n)$, a possibly
cubic dependence in $n=|V(G)|$ and quadratic in $q$ is asserted. To improve this dependence
we  use a more complex
algorithm that decomposes $G$ into low conductance pieces. We get,

\begin{theorem}\label{thm:52}
    Let $G_n$ be a unit-weighted graph that is known in advance.  There exists an online embedder that can embed $q$ points  from
    $X_G$ into ${\mathbb R}^2$ with  distortion $O(mD^{1/2}\log n \cdot q^2 +
    nD)$.
  \end{theorem}

 Theorems \ref{thm:27}, and \ref{thm:52}, are stated for the case that
 the number of points, $q$, is known in advance. For $q$ not known
  in advance, we could still prove the same results with an additional
  $\log q$ factor.   

For weighted graphs we prove the following. 
  \begin{theorem}\label{thm:weighted}
    Let $G$ be a {\bf weighted} graph with arbitrary weights, known in
    advance. There exists an embedder that can online  embed  $q$ points from
    $X_G$ and with distortion  $ O(n^3\cdot q^2)$, where $n=|V(G)|$. 
\end{theorem}
The idea in this case, following the strategy above, is to decompose
$G$ into low diameter pieces rather than low-conductance. We note that
all results above are 
independent of the aspect ratio.

  \vspace{0.3cm}
  \noindent
  {\bf Sub Classes of Tree Metrics: }
  
\noindent
  {\bf Ultrametrics: } 
We show that $q$ points taken from an ultrametric  can be
embedded online in low dimensional Euclidean spaces, with relatively
low distortion.

\begin{theorem}
  \label{thm:UM-line}
  Any $q$ points from $(X,d)$ that is UM can be online embedded into Euclidean ${\mathbb R}^k$  with
  distortion $O(k \cdot q^{2/k})$. Here we assume that $q$ (or a bound
  on it) is known in advance.
\end{theorem}
In particular, UM can be embedded
  with logarithmic distortion in logarithmic dimensional space, and with distortion
  $O(q^2), ~ O(q)$ into the line, plane, respectively.
We comment in Remark \ref{rem:12} in Appendix~\ref{A:UM} that  for $k=O(\log q)$ the leading $k$
factor in the theorem can be eliminated, resulting in a {\em constant}
distortion. We also note that to achieve a  constant distortion, logarithmic dimension
is needed even in an offline embedding.

The idea is to use condition (b) in the definition of
UM (Section \ref{sec:notation}), and show that a UM  can be
clusterable online in such a way that an online embedding is made
possible. For this we prove a more general result for appropriately
clusterable metric spaces in Theorem \ref{thm:cluster} 

\vspace{0.3cm}
\noindent
{\bf $\alpha$- Trees: } We prove that $\alpha$-trees can
be offline embedded into low dimensional Euclidian spaces with
relative low distortion. This holds for any $\alpha < 1$. 
      \begin{theorem}
        \label{thm:alpha-tree}
  Let $(X,d)$ be the metric that is induced by an $\alpha$-tree, with
  $\alpha < 1$.  Assume that the root is exposed first. Then, $(X,d)$ can be online
  embedded into Euclidean ${\mathbb R}^k$ with distortion $O(\frac{k \cdot
q^{2/k}\log q}{\log (1/\alpha)})$. Here we assume that $q$ is known in advance.
      \end{theorem}

The proof in this case is conceptually more  complex. We first
show that, if $\alpha$ that is relatively small w.r.t to $1/q$, and the
tree is exposed in a 'natural' order in which descendants
are exposed after their predecessors, a small distortion  online
embedding exists.  We then show how to deal with larger $\alpha$'s
and adversarial order. 
      
\section{Proofs}\label{sec:proofs}
We present here some of the proofs and ideas. Some of the  proofs appear in the
Appendix, but due to space limitations, some will appear in the final (or archive) paper.

\subsection{The 2-Paths game}\label{sec:2game}
\ignore{
  Let $P_a = [0,1]\times \{a\}$ and $P_b = [0,1]\times \{b\}$ equipped
with the
natural line metric $d$ on each. We define  the
metric space  $(P_a \cup P_b,\mu)$, to be:  for every $x,y \in P_c,
c\in \{a,b\}, ~ \mu(x,y) = d(x,y)$, and for $x \in P_a, y \in P_b,~
\mu(x,y)=1$.

The goal here  is to map $(P_a \cup P_b, \mu)$ to the unit square 
$S=[-1/2,1/2]^2
\subseteq \mathbb{R}^2$ in an online manner, where the image of the map
has small distortion w.r.t $\mu$. Namely,  we are given $q$
labeled points in $[0,1] \times \{a,b\}$, one at a time (with $q$ 
known or unknown in advance).
The crux is that we are restricted to map $[0,a],[1,a]$ to
$[-1/2,0],[1/2,0]$ respectively, and $[b,0],[b,1]$ to $[0,-1/2],[0,1/2]$
respectively, see Figure \ref{fig:110}.

For any two points $\alpha, \beta \in S$ let $\ell(\alpha, \beta)
= \ell_2(\alpha,\beta)$ be the Euclidian metric on $S$.
The result here is,
\begin{theorem} \label{MainTheorem}
  There exists an online procedure  that can embed $q$  points,  and achieve $dist(\phi) =
 O(q^2)$, when $q$ is specified in advance. If $q$ is not known ahead,
 then for any fixed $\epsilon < 1$ 
 $dist(\phi)=O(q^2\log^{2+\epsilon} q)$ can 
be achieved.
\end{theorem}
}

We describe here the online embedding procedure assuming
that $q$ is known. At the end of Appendix Section \ref{sec:A1} we comment on the
changes that need to be done if $q$ is
not known in advance. 

Recall that we have  two  unit length paths $P_a,P_b$.
 Our goal is to embed any $q$ points in
$P_a \cup P_b$ into 
$[-\frac{1}{2},\frac{1}{2}]^2$ with endpoints as shown in Figure \ref{fig:110} left
part. Instead, we will embed $P_a \cup P_b$ in the shifted scaled square $[-2, ~4q +2] \times [-2, ~4q +2]$, where $\phi((0,a)) = [-2,0), ~
\phi((1,a)) = (4q+2,0)$ and $\phi((0,b)) = (0,-2), ~
\phi((1,b)) = (0,4q+2)$ (see Figure \ref{fig:110} right part). The
original problem can be easily reduced to this one with the same
distortion.


\begin{figure}
\begin{center}
\includegraphics[scale = 0.7]{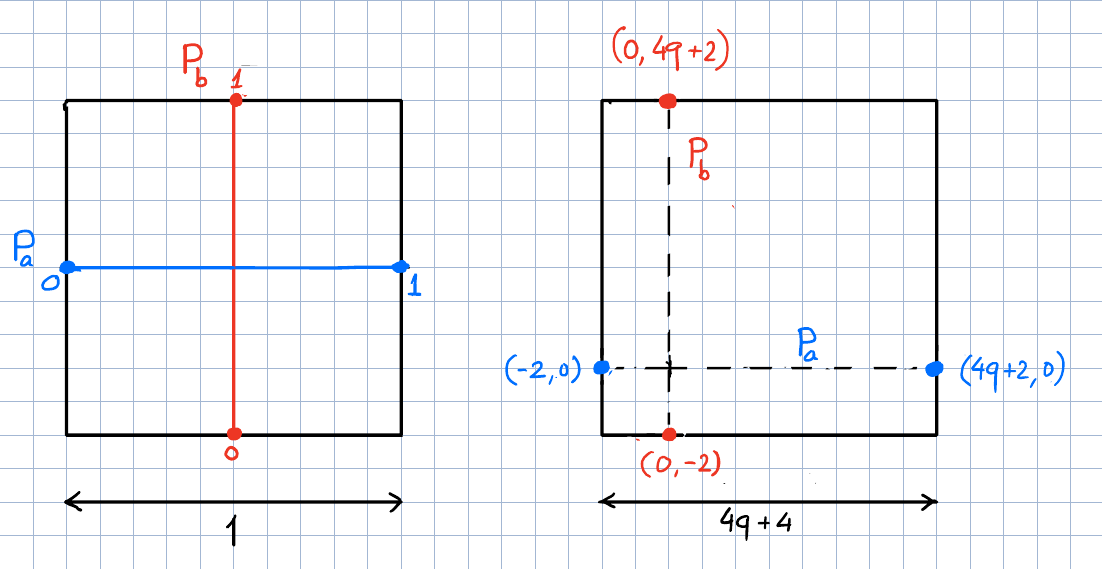}
\end{center}
\vspace{-0.4cm}\caption{Left figure: the horizontal path $P_a$ and vertical path $P_b$ should be
  embedded in the grid with endpoints as shown. Right figure:  the paths in the scaled shifted square}
\label{fig:110}
\end{figure}

\vspace{0.3cm}
\noindent
{\bf Algorithm}\\
We divide each line $P_a,P_b$ into $q+1$ closed subintervals
of length $\frac{1}{q+1}$,
$I_i = [\frac{i-1}{q+1}, \frac{i}{q+1}]\times \{c\}, ~ i=1, \ldots
q+1$. We refer to such subintervals as ``segments''  and denote
$I_i^c$ the subinterval $I_i$ of the $c$-''colored'' path.
We will pre-embed the boundary points of all segments linearly between the
fixed end-points of $P_a,P_b$, even if they are not
queried by the challenger  (that is, the adversary). Furthermore, before any query is made we will commit and embed
the whole first segment $I_1^a$  as straight line between its endpoints,
and $I_1^b$ along the path starting at $(0,-2)$ going straight to
$(4,-2)$, then going up to $(4,2)$, and finally to the intended
endpoint $(0,2)$, 
 see Figure \ref{fig:2} left part.  With this in mind,
every time a point $x \in I_i \subseteq P_c, ~ c \in \{a,b\}$ is
queried, we will embed {\em the entire segment} $I_i^c$ as a continuous curve 
of vertical and horizontal line segments, and 
embed $x$ linearly along this curve.  The whole point is to determine how the
image $\phi(I_i)$ is embedded.  This embedding will have the following properties.

\begin{enumerate}
\item The image of $\phi(I_i^a), ~ i=1 \ldots q+1$  is a continuous
  polyline along the grid lines. It  will be composed of at most $6$
  horizontal and vertical lines segments. Its 
{\em end points} are the predefined end points of $I_i^a$ in the
shifted scaled grid, $\phi(\frac{i-1}{q+1},a) = (4(i-1) -2, 0)$ and
$\phi(\frac{i}{q+1},a)= (4i-2, 0)$.  
We note that this does not mean that the
entire  
image of $I_i^a$ will be on that horizontal line. Similarly, the {\em
the end  points} of  $I_i^b$ are at 
$(0, 4(i-1) -2), (0,4i-2)$ respectively. 
  
    
\item The images of pairwise disjoint segments do not intersect. The
  image of a segment never intersects itself.
\end{enumerate}

\begin{figure}
\begin{center}
\includegraphics[scale = 0.7]{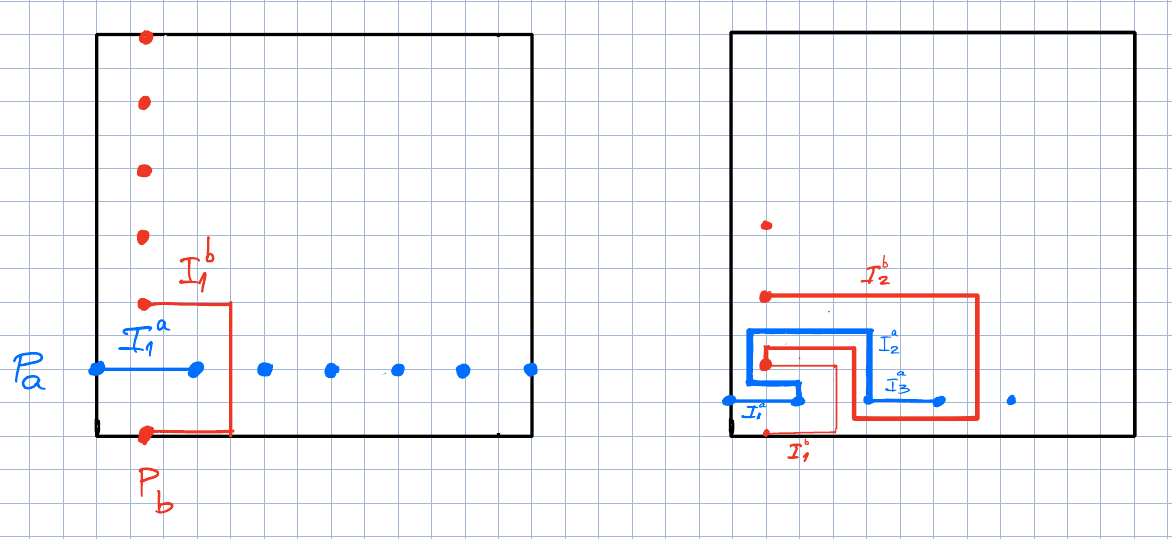}
\end{center}
\vspace{-0.4cm} \caption{Left part: the initial embedding; The dots are the end points of the
  segments, and  are embedded in advance even if not queried. The initial
  embedding starts with embedding of $I_1^a, I_1^b$ as shown.
Right part:   An example of an initial  query sequence $I_3^a, I_2^a,
  I_2^b:$ At the beginning, after $I_1^a,I_1^b$ are embedded,  all
  intervals are open except $I_1^a$ that is closed, and $I_2^a$ that
  is active. After the embedding of $I_3^a$ (making it closed), and
  the embedding of 
  $I_2^a$, the segment $I_2^b$ becomes active. Then, the embedding of $I_2^b$
  makes $I_4^a$ active. }
\label{fig:2}
\end{figure}

\ignore{
\begin{figure}
\begin{center}
\includegraphics[scale = 0.5]{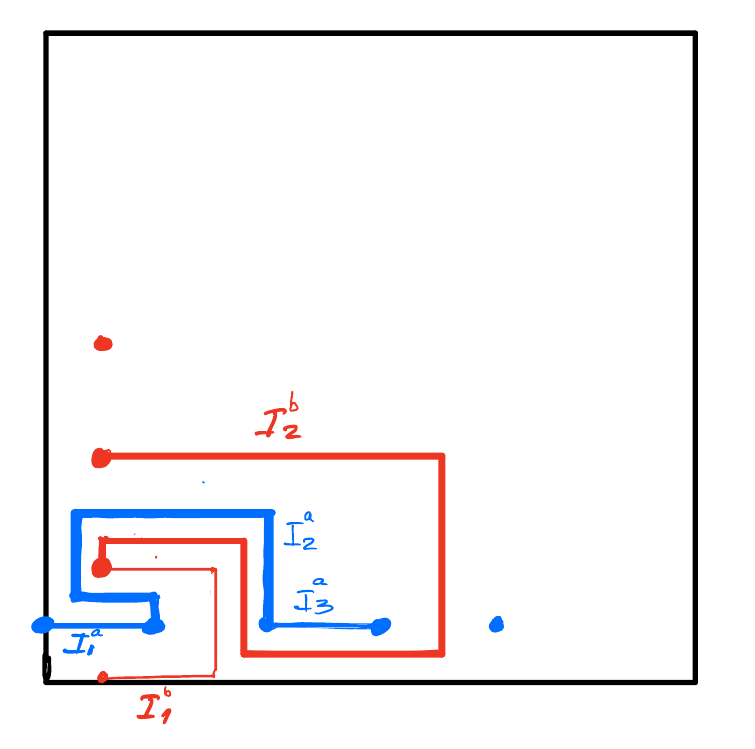}
\end{center}
\caption{
  An example of an initial  query sequence $[I_3,a], [I_2,a],
  [I_2,b]:$ At the beginning, after $I_1^a,I_1^b$ are embedded,  all
  intervals are open except $I_1^a$ that is closed, and $I_2^a$ that
  is active. After the embedding of $I_3^a$ (making it closed), and
  the embedding of 
  $I_2^a$, the segment $I_2^b$ becomes active. Then, the embedding of $I_2^b$
  makes $I_4^a$ active.}
\label{fig:3}
\end{figure}
}

At any step $i =1, \ldots, q$  several
non crossing segments of $P_a,P_b$ have already been embedded in the square.  In what follows we
define how to embed the interval $I$ in which the next point $x$ is in. 

\begin{definition}
At step $i$, a segment $I_j^c$ is denoted:
\textbf{Closed} if $\phi(I_j)$ is already embedded. 
\textbf{Active} if the embedding $\phi(I')$ of a segment $I'$ of
different color passes through the straight line 
between the end-points of $\phi(I_j)$, and 
\textbf{Open} otherwise. 
\end{definition}

Initially, $I_1^a$ and $I_1^b$ are closed, and 
 $I_2^a$ is active due to the fact that the image $\phi(I_1^b)$ passes
 through its endpoints (Figure \ref{fig:2} left part). All  other segments are open.
If a query to an open segment is made, this entire open segment can be
simply embedded as 
the shortest path between its corresponding, already embedded, endpoints. We show that an
active segment can also be mapped to a line in the grid  without crossing any closed
segment. The resulting image will be of length
$O(q)$, while crossing between the endpoints of exactly one open interval of the other
color. To
explain how this is done, assume that we need to map the active
segment $I \subseteq P_a$ (the reader is suggested to follow the embedding of
$I_2^a$ as depicted in Figure \ref{fig:2}, right part, or try to embed $I_4^a$
that is active after the initial depicted query-sequence).  By definition, $I$ is active since there
is an embedded segment   $I' \subset P_b$ whose image $\phi(I')$ 
passes in between the end points of $I$. We embed $I$ as follows:
$\phi(I)$  starts going vertically up, in parallel to $\phi(I')$
as far it can without intersecting $\phi(I')$. Then  it goes
horizontally left, in parallel to $\phi(I')$, meeting the line
$x=-1$. Then it bends again and goes vertically up, until it reaches the first open segment of $P_b$ where it
can reach the line $x=1$.  It now 
horizontally turns to the right, in parallel to $\phi(I')$ until it reaches
just above the 
correct end-point, at which point it just goes down to that endpoint.


Consider, e.g.,  the following sequence of queried intervals:
$I_3^a,~ I_2^a,~ I_2^b$. An embedding of it is shown in Figure
\ref{fig:2} right part. The idea is that in
every step $i \leq q$ there will be a suitable open interval (not yet
embedded) allowing the current active interval to 'go through'.
Further details and the formal description of the algorithm are in 
Appendix \ref{sec:A1}. 

At the end of the formal proof we also describe
what  to do  when $q$ is not known in advance.

\subsection{Applications to 'solid' graph metrics}\label{sec:app}
The proofs of Theorems \ref{thm:27}, \ref{thm:52} and \ref{thm:weighted} , appear in 
Appendix \ref{sec:A3}. For this we also introduce and use a generalization of the
$2$-paths game to $k$-paths in Appendix \ref{sec:kk-inter}.

\subsection{Ultrametrics}\label{sec:trees}

\ignore{
\subsection{ultrametrics}\label{sec:4.1}
A metric space $(X,d)$ is ultrametric, abbreviated UM,  on a finite set $X$ if for
every three points $x,y,z \in X$, $~d(x,z) \leq \max \{d(x,y),
d(y,z)\}$. It is folklore (and easy) that the following conditions,
each characterize UM.

(a) - $d$ is the metric induced on the set of a leaves of a weighted  tree with
a special vertex $r$, and such that $d(r,x)$ is the same for every $x
\in X$.

(b) - $X$ is a subset of the vertex set of a weighted undirected graph
$G =(V,E)$, and such that $d(x,y) = \min_{P_{x,y}} \max_{e \in
  P_{x,y}} w(e),$ where $P_{x,y}$ denotes any path from $x$ to $y$ in
$G$.

By (a), UM is a special case of a tree metric. On the other
hand, it is quite obvious that any metric on a finite set $X$ of size
$n$ can be embedded (offline) into UM with distortion at most
$n-1$. This is so as taking the graph $G=(X,E)$ for which $w(x,y) =
d(x,y)$ represent an UM $d'$ on $X$ (by (b) above). Since 
the weight of the heaviest edge in any path is at most  the
length of the path and at least this length divided by $n-1$,  we
conclude that $d/ (n-1)   \leq d' \leq d$.  Moreover, this is best possible as shown by the
metric on the unit weighted path of length $n$.

We show below that $q$ points from a UM can be {\em online} embedded
 with distortion $O(k\cdot q^{2/k})$ in $\ell_2^k$ (that
is in ${\mathbb R}^k$ with Euclidean distance), when $q$ is known in
advance. In particular, it can be embedded into the line with
distortion $O(q^2)$. This implies an exponential lower bound on the
distortion of online embedding a general metric, or even a general
tree metric, into UM.

\begin{theorem}
  \label{thm:UM-line}
  Any $(X,d)$ that is UM can be online embedded into Euclidean ${\mathbb R}^k$  with
  distortion $O(k \cdot q^{2/k})$. In particular, UM can be embedded
  with logarithmic distortion in logarithmic dimensional space, with distortion
  $O(q^2), ~ O(q)$ into the line, plane, respectively.
\end{theorem}
We comment in Remark \ref{rem:12} that  for $k=O(\log q)$ the leading $k$
factor in the theorem can be eliminated, resulting in constant distortion.
}

 Theorem \ref{thm:UM-line} will follow from the more general Theorem
\ref{thm:cluster} below whose proof is in Appendix \ref{A:UM}, and
Observation \ref{obs:82} (Appendix \ref{app:UM}).
 We need the following,
\begin{definition}
  [Clusterable metric space] \label{def:cluster}
  A finite metric $(X,d)$ is $(\gamma, \epsilon)$- clusterable, if there is a
  non-trivial partition $X = C_1 \cup \ldots \cup C_{\ell}$, called ``clusters'', such that:
  \begin{enumerate}
  \item For every $x,y$ in different clusters $d(x,y) \geq \D(X)/\gamma$.
       
      \item For every cluster $C_i$, $\D(C_i) \leq \epsilon \cdot \D(X)$.
  \end{enumerate}
  A finite metric $(X,d)$ is $(\gamma, \epsilon)$- recursively-clusterable if
   it is either $(\gamma, \epsilon)$-clusterable to singletons, or it is
  $(\gamma,\epsilon)$-clusterable and each non-singleton class of it is
  $(\gamma,\epsilon)$-recursively-clusterable.
\end{definition}

\ignore{
\subsection{Clusterable metrics}\label{sec:clusterable}

\begin{definition}
  [Clusterable metric space] \label{def:cluster}
  A finite metric $(X,d)$ is $(\gamma, \epsilon)$- clusterable, if there is a
  non-trivial partition $X = C_1 \cup \ldots \cup C_{\ell}$, called ``clusters'', such that:
  \begin{enumerate}
  \item For every $x,y$ in different clusters $d(x,y) \geq \D(X)/\gamma$.
       
      \item For every cluster $C_i$, $\D(C_i) \leq \epsilon \cdot \D(X)$.
  \end{enumerate}
  \end{definition}
\begin{definition}[recursively-clusterable] \label{def:rec_cluster}
  A finite metric $(X,d)$ is $(\gamma, \epsilon)$- recursively-clusterable if
   it is either $(\gamma, \epsilon)$-clusterable to singletons, or it is
  $(\gamma,\epsilon)$-clusterable and each non-singleton class of it is
  $(\gamma,\epsilon)$-recursively-clusterable.
\end{definition}
The definition is meaningful when $\epsilon << 1/\gamma$. In that case $(X,d)$
is $(\gamma,\epsilon)$-clusterable if $X$ can be partitioned into clusters
with distances between clusters that is much larger w.r.t the inner distance within 
clusters. A UM is not sufficiently clusterable in general. Observation
\ref{obs:82} (Appendix \ref{app:UM}) bridges this gap.

Note that  if $(X,d)$ has aspect ratio at most
$\gamma$ then it is $(\gamma,0)$-clusterable into singletons, hence every unit
weighted graph metric is $(n-1,0)$-clusterable into
singletons. 

\begin{definition}
  [recursively-clusterable] \label{def:rec_cluster}
  A finite metric $(X,d)$ is $(\gamma, \epsilon)$- recursively-clusterable if
   it is either $(\gamma, \epsilon)$-clusterable to singletons, or it is
  $(\gamma,\epsilon)$-clusterable and each non-singleton class of it is
  $(\gamma,\epsilon)$-recursively-clusterable.
\end{definition}
Clearly if $(X,d)$ is recursively clusterable, then any submetric
of it is  recursively clusterable with the same parameters. Hence
we may refer  to the cluster-structure of $(X,d)$ at any intermediate
step after partially exposing  $X$.
}

\begin{theorem}
  \label{thm:cluster}
  Let $(X,d), ~ |X|=q$ be $(\gamma, \epsilon)$-recursively-clusterable, with
  $\epsilon < \frac{1}{2\gamma \sqrt{k} q^{1/k}}$ then $(X,d)$ can be online
  embedded into Euclidean ${\mathbb R}^k$ with distortion
  $O(2\gamma^2\sqrt{k}q^{1/k})$.  We assume here that $q$ is known in advance.
\end{theorem}
The definition is meaningful when $\epsilon << 1/\gamma$. In general, a
UM is not clusterable well enough to directly use the theorem. Observation
 \ref{obs:82} (Appendix \ref{app:UM}) bridges this gap.


\subsection{$\alpha$-Tree Metrics - Theorem \ref{thm:alpha-tree} }
\ignore{
\subsection{Other clusterable metrics}\label{sec:rooted-trees}
      A natural example of a clusterable metric is given by a rooted
      tree whose weights exponentially decrease on every root-leaf
      path. This motivates the following definition.

      \begin{definition}
        [$\alpha$- tree metric]
        Let $T=(V,E)$ be an edge-weighted  rooted tree with root $r
        \in V$. Let $0 < \alpha < 1$. We say that $T$ is an $\alpha$- tree
        ($\alpha$-tree) if for any root-to-leaf path $(r=v_0, v_1,
        \ldots ,v_\ell)$, $~ w(v_i,v_{i+1}) \leq \alpha \cdot
        w(v_{i-1},v_i)$.

        For a rooted tree $T$ and a directed path as above, we say
        that $v_i = father(v_{i+1})$. We also say that it is an ancestor of $v_{i+j}, ~
        j\geq 1$, and descendent of $v_{i-j}, ~ j\geq 1$.
      \end{definition}

      \begin{remark}
        \label{rem:14}
        \begin{enumerate}
        \item         We note that $\alpha$-tree is a generalization
          of Bartal HST \cite{Bar}, since we do not require that all edges going
          out from a node have the same weight.
        \item We also note that the shortest-path metric that is
          defined on $V$ by an $\alpha$-tree has natural recursive
          clusters formed by the subtrees of $T$. However, for any
          $\alpha$ this does not defines, in general, a
          $(\gamma, \epsilon)$-recursively-clusterable metric for
          $\gamma, \epsilon$ as needed by Theorem
          \ref{thm:cluster}. E.g., for some $\gamma < 1$, consider a
          star rooted at its center with leaves in $[n-1]$ of
          distances $d(center,i)= \gamma^i$.

          While $d_T$ that is
          defined by an $\alpha$-tree is not clusterable as in Section
          \ref{sec:clusterable}, we note that if the root is revealed
          first the subtrees structure of $T$ are well defined online.
          Indeed assuming this, we have the following theorem.
        \end{enumerate}
      \end{remark}
      \begin{theorem}
        \label{thm:alpha-tree}
  Let $(X,d)$ be the metric that is induced by an $\alpha$-tree, with
  $\alpha < 1/5,$ and assume that the root is exposed first. Then, $(X,d)$ can be online
  embedded into Euclidean ${\mathbb R}^k$ with distortion $O(k \cdot
q^{2/k}\log q/\log (1/\alpha))$.
      \end{theorem}
    }
 
 Let $T=(X,E)$ be an $\alpha$-tree  and $d$ the corresponding
 metric. Assume that $X$ is exposed in an arbitrary order $x_0=root$,
 $x_1, \ldots x_{q-1}$, with the root being first.
 We denote by  $T^{(t)} $ the subtree that is defined by
$X^{(t)} = \{x_i,~ i \leq t\}$ and $d^{(t)}$ the metric induced on  it.
The exposure order is a {\bf natural order} on $V(T)$ if every point  $v$
appears before any $u \in T_v \setminus \{v\}$. That is, predecessors
appear before their ancestors.

If $u$ is a predecessor of $v$ then
 $\D(T_v) < 2\alpha \cdot \D(T_u)$.  Hence, assuming a small enough
 $\alpha$, we could wish to embed $d_T$ online and so to maintain the
 structure of subtrees. Namely, similarly to the embedding of UM, to
 embed each subtree $T_v$ rooted at $v$ 'near' the embedding of $v$,
 using $v$ as a 'fat point' (as described in Appendix, Section 
 \ref{A:UM}).
 However,  we can do this only if  the tree is
 exposed in a {\em natural order}. In addition,  even in this 'simple'
 case, we would need $\alpha$ to be quite small.

 To prove the theorem we use three independent steps. 
  Theorem \ref{thm:natural_order} asserts that, in the 'simple' case
 where the exposure  is in the {\em natural} order {\em and}
 $\alpha$ small enough, a relatively easy embedding exists. Then we
 show, in Theorem \ref{thm:tree-to-natural}, how to {\em online} embed
 an $\alpha$-tree into another $\alpha'$-tree where the possibly
 adversarial exposure order w.r.t the input becomes natural w.r.t the
 host tree. Lastly, in Theorem
 \ref{thm:alpha-to-beta}, we show that for natural order exposure and $\beta$ arbitrary small, 
  an $\alpha$-tree can be online embedded into a
 $\beta$-tree  with small distortion.
 Finally we can combine the
 above to prove Theorem \ref{thm:alpha-tree}. The proofs of Theorems
 \ref{thm:natural_order}, \ref{thm:tree-to-natural} and 
 \ref{thm:alpha-to-beta} are in Appendix Section \ref{sec:A-alpha}.

\ignore{

\begin{remark}\label{rem:32}
      For $x=x_t$ and $y \in X^{(t-1)}$ let $T^{(t)}(x,y)$
be largest subtree of $T$ that contains $x,y$. We note that
$T^{(t)}(x,y)$ is rooted at some $v$ that is not necessarily exposed
by time $t$, but that is well defined by $d(x,x_j), ~ j <t$. Hence if
$v \neq x$ is not in $X^{(t)}$ 
we may re-enumerate the order and assume that $v$ is exposed at time
$t$ while $x$ is exposed at $t+1$, resulting possibly in a sequence of
$2q$ exposed vertices that we are going to account for. Hence in all
what follows we assumed that $v$ as above is always exposed before $x$, and it will be embedded too even if $v$ is not a part of the
exposed sequence at all. Abusing notation we refer to this new order
as $x_0, \ldots ,x_{q-1}$ (redefining the new number of exposed points
that might $2q$ by $q$).
 \end{remark}
 
  We start with the following simple observation.
  
  \begin{observation}
  \label{obs:93}
 \begin{enumerate}
   \item The subtree
      $T^{(t)} $ that is defined by $X^{(t)}$ is a also
      $\alpha'$-tree, for $\alpha' \leq \frac{\alpha}{1-\alpha}$ (as
      some paths in $T$ may becomes an edges in $T^{(t)}$).
      
 \item  For a subtree $T_v$ rooted at $v$ and with $u=  father(v)$ in $T$, $\D(T_v) \leq  \frac{2\alpha}{1-\alpha} d(u,v)$.
\end{enumerate}
\end{observation}
}

\begin{theorem}
  \label{thm:natural_order}
  Let $(X,d)$ be the metric that is induced by an $\alpha$-tree, with
  $\alpha <  \frac{1}{32\sqrt{k}q^{1/k}}$. Assume that $X =
  \{x_0, \ldots ,x_{q-1}\}$ is exposed in the natural order and that
  $q$ is known in advance.  Then, $(X,d)$ can be online
  embedded into Euclidean ${\mathbb R}^k$ with distortion $O(\sqrt{k}q^{1/k})$. 
\end{theorem}

\begin{theorem}
  \label{thm:tree-to-natural}
  Let $(X,d)$ be the metric that is induced by an $\alpha$-tree 
  and suppose that $root$
  appears first. Then $(X,d)$ can be online embedded into an
  $\frac{\alpha}{1-\alpha}$-tree $\tilde{T}
  = (X',E')$ inducing a metric $\tilde{d}$ that dominates $d$ and
  expand it by at most $\frac{1+\alpha}{1-\alpha}$. Further, the order
  of exposure w.r.t $\tilde{T}$ is natural.
\end{theorem}

\begin{theorem}
  \label{thm:alpha-to-beta}
  Let $T$ be an $\alpha$-tree and $\beta <
  \alpha$. Assume that  $T$ is exposed in the natural order, then $d=d_T$
  can be online embedded in the natural order into a $\beta$-tree
  $\tilde{T}$ with distortion $\frac{\log
  1/\beta}{\beta \cdot \log (\frac{1-\alpha}{\alpha})}$. 
\end{theorem}

Combining the above theorems, the proof of Theorem
  \ref{thm:alpha-tree} follows: 
\begin{proof}
  Let $T = (X,E)$ be an $\alpha$-tree with root $o$. Let $d_T$ the
  corresponding metric. We first use Theorem \ref{thm:tree-to-natural}
  to online embed $X$ into a tree  $T_1$ that is a
  $\alpha_1=\frac{\alpha}{1-\alpha}$-tree rooted at $o$ where the
  exposure order on $T$ induces a natural order on $T_1$. Theorem
  \ref{thm:tree-to-natural} implies that the resulting metric $d_1$
  distorts $d_T$ by  $\frac{1+\alpha}{1-\alpha}$.

  Next we set $\beta = \frac{1}{32\sqrt{k}q^{1/k}}$ which is the
  parameter needed for Theorem \ref{thm:natural_order}. We use Theorem \ref{thm:alpha-to-beta} to
  online embed $d_1$ into a $\beta$-tree $T_2$ where the original
  exposure order is the natural order for $T_2$.  The resulting  metric
  $d_2$ distorts $d_1$ by 
  $O(\frac{\sqrt{k} q^{1/k} (\frac{\log q}{k}+ \log k)}{\log
    (\frac{1}{\alpha})}) = O(\frac{q^{1/k}}{\log (1/\alpha)}  \cdot
  (\frac{\log q}{\sqrt{k}}+ \sqrt{k} \log k)) $. 

Finally, we can now online embed $d_2$ using Theorem
\ref{thm:natural_order} into ${\mathbb R}^k,$ resulting in a $\ell_2$ metric that
distorts $d_2$ by $O(\sqrt{k}q^{1/k})$.  Multiplying the 3 terms 
 implies that the total distortion is  $O(
\frac{q^{2/k}}{\log (1/\alpha)} \cdot (\log q + k \log k))$.
\end{proof}

We comment that, here too, like in Remark  \ref{rem:12} in the Appendix, we could use  spherical codes instead of the
grid in Theorem  \ref{thm:natural_order}. This will omit the
$k$ factor from the distortion estimate, for $k = \log q,$ but this is less significant
in this case, as for $\alpha = \Omega(1)$ (e.g., $\alpha=1/2$) we
would still get a $\log q$ distortion even for logarithmic dimension.

\newpage
\bibliography{ilan_lipi1_test}

\newpage
\appendix

{\bf \LARGE{Appendix}}
\section{Formal proof of the 2-paths game}\label{sec:A1}
Recall the definition of 'open', closed' and 'active' intervals.  We will keep the following invariant.
\begin{inv} \label{ActiveInterval}
At any step of the algorithm, there is exactly one active
interval. Additionally the active interval alternate between $P_a$ and
$P_b$.
\end{inv}

Following is the formal definition of the algorithm:

\noindent
{\bf Algorithm}
On input query $x\in I_i \subseteq P_c$ do:
\begin{enumerate}
    \item If $I_i$ is closed, then $\phi(x)$ is already defined.
    \item If $I_i$ is open, then map the interval to the straight line between its
 endpoints.
    \item Otherwise $I_i$ is active. We assume in what follows that
      $c=a$. The case $c=b$ is similar.
    Let $I_j$ be the last previously active interval, which by the
    invariant above $I_j \subset P_b$, and let $I_k^b$ be the first open segment of 
$P_b$.
    Let $s_i=4i -2$, namely, $(s_i,0)$ is the starting end-point of $I_i$.
    Let $s_j= 4j-2, ~s_k = 4k-2 $, that is $(0,s_j)$ and $(0,s_k)$ are
    the starting point of $I_j,I_k$ respectively.

    Then $I_i^a$ is linearly mapped into the following sequences of line segments, 
    $(s_i,0) \rightarrow
    (s_i+1, 0) \rightarrow
  (s_i+1, s_j + 3) \rightarrow
    (-1, s_j + 3) \rightarrow
    (-1, s_k + 1) \rightarrow
    (s_i+4, s_k+ 1) \rightarrow
    (s_i+4, 0).$
    
    See Figure \ref{fig:4} ($I_i^a$ is embedded into the black doted
    line).     The case when $c=b$ is symmetric.

\begin{figure}[h]
\begin{center}
\includegraphics[scale = 0.7]{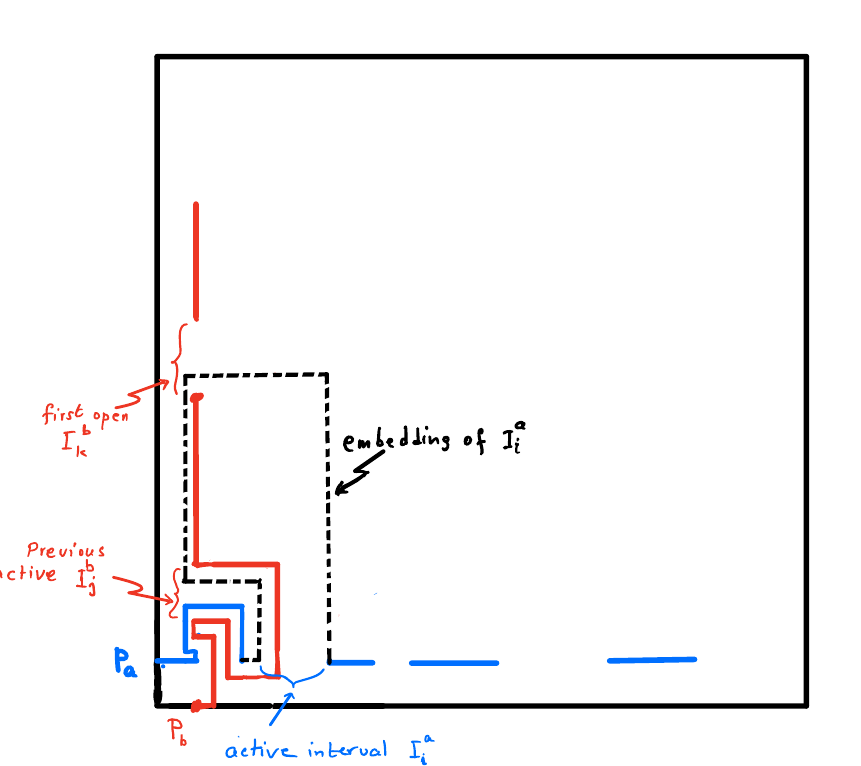}
\end{center}
\caption{The embedding of an active interval $I_i^a$ is depicted as s
  dashed line.}
\label{fig:4}
\end{figure}
\end{enumerate}

\ignore{
Before making the formal claims, consider the example as shown in
Figure \ref{fig:2}), of embedding the sequence .........
Since $I_3^a$  is open, it can be mapped to the straight line between its 
endpoints. $I_2^a$ is active (due to the embedding of $I_1^b$), and
will bypass the blockage, as formulated above (see 
 As a result, since 
$I_2^aa$ passed in between the endpoints of $I_2^b$, $I_2^b$ becomes 
active (and $I_2^aa$ becomes closed). Thus on the next query to
$I_2^b$ we will need 
to make another detour, which will be symmetric to the previous one, but because
 $I_3^a$ is closed, the path will pass in between the endpoints of  $I_4^a$, 
causing it to be the active interval. See Figure \ref{fig:3}.
}

\begin{lemma}\label{numberOfQueries}
  Invariant \ref{ActiveInterval} is maintained after every step of the
  algorithm as long as at most $q$ queries are made.
\end{lemma}
\begin{proof}
  The proof is formally by induction on the number of actual queries
  $\ell$ that are made (for $q$ fixed and known in advance). The base
  case is obvious.

  Suppose that the lemma holds, and that  $x_{\ell} \in
  I_i$ as described in the algorithm above, for $\ell \leq q$. 

  The only part of the algorithm that can fail is the 3rd case
  where it needs to find an open interval $I_k$. Assume the active interval
  is a a-interval. Since at the beginning there are $q$ open intervals of
  $P_b$, and each query can close only one interval, there is  $I_k\subseteq
 P_b$ that is open. Similarly when the active interval is a b-interval,
  there are $q-1$ open intervals of $P_a$, and at least one query was made to
  $a$ in order to switch the active interval, so there is an open
  interval $I_k\subseteq
  P_b$.

  Assume that $\ell \leq q$. Then, the active interval changes only in the case when the query
  point $x$ is in the current active interval, namely, when the 3rd
  case  in the algorithm happens. Assume w.l.o.g, as
  in the description in the algorithm, that  $I_i \in P_a$ is this active
  interval. In this case $I_i$ becomes closed,  $I_k \subseteq P_b$ becomes
  active, which proves the invariant. The proof for $I_i \subseteq P_b$ is
  similar.
\end{proof}

\begin{lemma} \label{noIntersection}
At any step $\phi$ is injective. Namely, 
 for every pair 
of embedded  points $x\neq y$, $\phi(x)\neq\phi(y)$. $\qed$
\end{lemma}

\begin{proof} 
[Proof of Theorem \ref{MainTheorem}]
By Lemma \ref{noIntersection} and the fact that the grid is integral, the minimum
distance between $x\in P_a$ and $y\in P_b$ is $1$.  Additionally, as
each subinterval $I_i \in P_c, ~ c \in \{a,b\}$ is stretched to a path
of integral length, and the queried points are mapped linearly along
the subinterval, there is no contraction in the resulting embedding.

Every interval $I_i\subseteq P_c$ is embedded into a union of at
most six horizontal or vertical $O(q)$-long straight line segments. Thus the  expansion is
$max_{x,y}\frac{l(x,y)}{\mu(x,y)}\leq\frac{O(q))}{1/(q+1)}=O(q^2)$,
implying a  $O(q^2)$ distortion.

If  $q$ is not known in advance, instead of splitting
$P_a,P_b$  into 
uniform intervals, we split the paths into infinitely many
intervals of  decreasing lengths.
We define a converging series $f(x) = \frac{c}{x log^{1+\epsilon/2}x}$, for $c$
 such that $\Sigma_1^{\infty}f(x)=1$. $f(i)$ will describe the 
percentage of the path length that  the $i$'th interval contains. 
The algorithm will remain exactly the same, except for the definition of the 
3'rd case. Because there isn't an integral grid, we will translate the addition 
of constant in the path coordinate, to an addition of a percentage of the 
interval length (for example $[s_i+1, s_j + 3]$ is translated into 
$[s_i+\frac{1}{4}f(i), s_j + \frac{3}{4}f(j)]$).

Lemma \ref{numberOfQueries} and \ref{noIntersection} remains the same, but we 
adapt the claim of \ref{numberOfQueries} to note that after $q$ queries, the 
active interval is at most $I_{q+1}$,
so $dist(\phi)\le f(q+1)^{-2}=O(q^2 log^{2+\epsilon} q)$.
\end{proof}

For the applications in next sections, we note that by scaling, if we embed  two lines
of length $l$ (rather than $1$), of pairwise distance $d$ (instead of $1$), into a $s\times s$ grid  (rather
then $\theta(q) \times \theta(q)$), 
Then we get  $O(\frac{qs}{l})$ 
expansion and $O(\frac{dq}{s})$ contraction.

\section{k-paths intersection}\label{sec:k-inter}\label{sec:kk-inter}
 We define below a generalization of the $2$-intersection that we call
"k-paths-intersection-gadget". This will allow us to use the algorithm
for that problem as a black box in online embedding of solid graph metrics.

Let  $X_k=[0,l]\times [k]$  be a disjoint union of $k$ paths of length
$l$, each equipped with the line metric $\ell$. The metric 
 $(X_k,\mu)$ is defined by: for each $x, 
y\in[0,l]\times[i], ~ \mu(x,y)=\ell(x,y)$ and for  
$x\in[0,l]\times[i], ~ y\in[0,l]\times[j], ~ i \neq j, ~ ~ \mu(x,y) = d$. Namely, the
 distance between points on different paths is $d$.
The goal is to online embed $q$ points from $X_k$ into the square
$[0,s]^2$ where the ends points are fixed in advance in an arbitrary
way that is presented to the embedding algorithm.  For our use we assume that $q >> k, l,d$.

%

   Formally, we assume that each side (edge) of
the square contains $k+2$ uniformly spaced points (including the
corners). We refer to the set $T$ of the $4k$ {\em internal} points on the sides (that is,
excluding the corners) as ``Terminals''. The {\em input} will be an
injective mapping $t: \cup_{i=1}^k \{(0,i),(l,i)\} \mapsto T$ that
specifies for each $i \in [k],$ the terminals that are  the
images of the endpoints of the $i$th path. This assignment is known to the algorithm
  in advance. 
The embedder must map the endpoints of each path, $[0, i]$ and $[l,i]$
 to the corresponding terminals $t(*)$.

 We refer to this embedding problem as the ``$k$-paths intersection''
 with parameters $l,d$, that is, the length of the paths and
 pairwise distance between paths, and $s$ which is the target box size in
 ${\mathbb R}^2$. We assume (for our applications) that $d,l \geq 1$, although
 the whole procedure is invariant under scaling.

\begin{theorem}
    There is an online procedure that  {\bf online} embeds  $q$
    points of $X_k$ into ${\mathbb R}^2$ with expansion $O(\frac{sq}{l})$ and 
contraction  $O(\frac{dkq}{s})$, implying a total distortion of 
$ O(\frac{dkq^2}{l})$.  The number of points, $q,$ is  known in advance.
   \end{theorem}
   \begin{proof}
   In order to define the online embedding  we  first define a scaled
   grid $[0,3(3k+1)]^2$ and define the embedding therein. Later,
   in order to embed in a $[s]^2$ box we just scale the grid by
   $a_1=\frac{s}{3(3k+1)}$ which will scale the expansion by $a_1$ and
   scale down the contraction by $a_1$ resulting in the same overall
   distortion.  We do not optimize the result here as
   this  gadget is mostly needed for our applications, and is less
   interesting on its own.

   The terminals $T$ will be the points
   $\cup_{i=1}^k \{(3i, 0), (3i, 9k+3), (0, 3i), (9k+3,
   3i)\}$.  The indices $i \in [k]$ in this definition
   do not necessarily correspond to the endpoints of the $i$-path, as
   the mapping of endpoints is arbitrary.
   
   Given the mapping of the endpoints to terminals, we define now a
   set of paths $\{\tilde{P}_i\}_1^k$ in the grid corresponding to the
   intended line metrics. First, for paths for which one terminal is
   on the vertical edge and the other is on an horizontal edge, e.g.,
   w.l.o.g. $x=(3i, 0)$ an $y$ is $(0, 3j)$.  The path will be
   $(3i,0) \rightarrow (3i, 3j) \rightarrow (0, 3j)$.

    Once this is done we declare all rows and all columns that contain 
    a segment that is already assigned as ``used''. We also declare
    all rows and all columns that contain an assigned  terminal 
    as used. This define at most $2k$ rows and $2k$ columns
    as used.

    Next, we define the intended paths for these path with endpoints
    both on a vertical edge or a horizontal edge (not necessarily the
    same). This is done as follows:  for each remaining path in an
    arbitrary order do:
    
    Let $x,y$ be the assigned terminals of some path $P_{\ell}$ to be
    currently defined. 
\begin{itemize}
    \item If $x$ and $y$ are on the same bottom edge of the square,
      i.e.,   $x$ is $(3i, 0)$ and $y$ is
    $(3j, 0)$, for some $i,j$. 
    Let $r\in[k]$ be such that the $3r$'th row is yet unused (there will
    be such $r$).
        The path will be 
    $(3i, 0)\rightarrow
    (3i, 3r)\rightarrow
    (3j, 3r)\rightarrow
    (3j, 0)$, and the $3r$'th row is now marked as used. A analogous
    definition is done for $x,y$ on the top edge or on the same vertical edge, using a
     previously unused column instead of row.
    \item If $x$ and $y$ are on opposite edges e.g., for $x=(3i,0)$ and
      $y=(3j,9k+3)$ the path will start as before, going to the first
      $3r$ unused row, and finally going
      up from $(3j,3r)$ to $(3j,9k+3)$, see Figure \ref{fig:5_5}.
\end{itemize}

\begin{figure}
\begin{center}
\includegraphics[scale = 0.7]{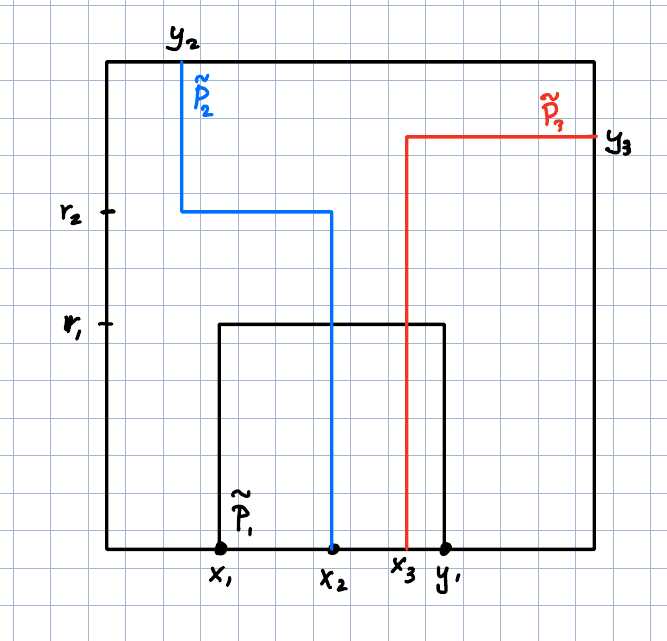}
\end{center}
\caption{The three types of paths in the grid. The red path with
  terminal in horizontal and vertical edges is allocated first.}
\label{fig:5_5}
\end{figure}

\begin{claim}
Let $\tilde{P}_i, ~ i=1, \ldots ,k$ be the paths that corresponds to
the terminal mapping above. Then, any two path meet in at most two
 integral points.  Further, any point in the grid belongs to at most two paths, and if two path
share a point, one passes through it vertically, and the other horizontally.
\end{claim}
\begin{proof}
By definition, each path is composed of  at most three line
segments. More
over, there are $3k$ rows and $3k$ columns of type $3i, ~ i \in
[k]$. After the first stage, at most $2k$ of these rows, and at most
$2k$ of these columns are declared used. Thus leaving at
least $k$ unused rows, as needed for the at most $k$ paths whose
endpoints lie on the horizontal edges, and with the same argument for columns.  

In addition at most one path travel at any given
row or at any given column of the grid. Hence paths can meet only if one is horizontal and
the other is vertical. Further, this can happen only at integral points.
\end{proof}
 Now the online embedder attempts to use $\tilde{P}_i$ for $P_i$ where
 points are mapped linearly along $\tilde{P}_i$. 

At each intersection (these are known in advanced, as the paths
$\tilde{P}_{i=1}^k$ are defined beforehand), we  use
"1-intersection" gadget to resolve the conflict, with a square of size $2$
centered at the intersection integral point. Note that no two
1-intersection gadgets intersect in the grid.

{\bf Distortion analysis:} {\em Expansion: } Since every $\tilde{P}$
is of length at most $18k+6$ and $P$ that is mapped to $\tilde{P}$
embeds linearly in $\tilde{P}$ then the expansion between points on
$\tilde{P}$ is bounded by $O(k/l)$. Moreover, the embedding of points
in a path $P$ may deviate from its mapped line $\tilde{P}$ only in
$1$-intersection gadgets. Each such expansion adds additively to the
total expansion. Further, for  two points on the
same path that are in an intersection gadget, their distance $l'=\Omega(l/k)$ could
increase by a factor $O(\frac{q}{l'}) = O(\frac{qk}{l})$ (see last par. in
Section \ref{sec:2game}). This results in a total $O(\frac{qk}{l})$
expansion inside a path.  For two points not on the same path the
original length is $d$ while the resulting length is $O(k)$ implying
an $O(k/d)$ expansion for such points. Hence the overall expansion is
$O(\frac{qk}{l})$.  Scaling this embedding to a $[s]^2$ box (by
multiplying by $a_1$, see first paragraph in the proof) results in 
$O(\frac{qs}{l})$ expansion.

{\em Contraction:} 
For two points on the the same line $P$ the contraction is determined
by the 1-intersection regions. For such two points the contraction is
$O(dq)$ (again see last par. in Section
\ref{sec:2game}). Otherwise, for two points of which at least one is not
in a $1$-intersection gadget, the distance may become $O(1)$ while the
initial distance is $d$. Thus the total contraction is bounded by
$O(dq)$, and once scaled down by $a_1$ it becomes $O(\frac{dqk}{s})$, as
claimed in the Theorem. 
\end{proof}

\section{Proof of Theorems \ref{thm:27}, \ref{thm:52} and \ref{thm:weighted}}\label{sec:A3}  
Here we consider the metric space $X_G$ for connected {\em unit} weighted graph
$G$. The graph $G$   is known in advance. 
 
\begin{proof}[Proof of Theorem \ref{thm:27}]
In order to online map $q$ labeled points from $X_G$  into the
plane, we first embed {\em offline,  part} of $G$ as
follows.  We subdivide each edge $(u,v) \in E$ into a $3$-path,
$(u,u',v',v)$ where $distance(u,u')  = distance(v,v') = 1/4$ and
$distance(u',v')=1/2$. We will do a preprocessing where we commit to
an embedding of all edges $(u,u')$ that are embedded  as straight lines, as
follows.

Let $\{v_i\}_{i=1}^n$ be a fixed enumeration of the vertices,  
$deg_i$  the degree of $v_i$,  $x_i = \Sigma_{j=1}^{i-1}{deg_j}, ~
i=1, \ldots ,n$, and $x_1=0$. 
We first embed $V(G)$ by $\phi(v_i) = (x_i,0)$.
Further, for each $v_i \in V$ we use an arbitrary enumeration of the
edges $\{(v_,u) \in E\} = (e_1, \ldots , e_{deg_i)})$, and embed the
first $1/4$-length part of $e_j=(v_i,v'_j)$ as a straight vertical line with end
point $\phi(v'_j) = (x_i+j-1, 1), j=1, \ldots deg_i$.

Now, in the {\em online} part,  any point in  $(u,u')$ is
already determined, and for the $(u'v')$ edges, we use a
"$m$-intersection" gadget of 
size $m$ to connect the endpoints of all edges, where
all $\phi(v')$ are below the box of the $m$-intersection. 
The parameters of the 
"$m$-intersection" are $s=m, ~d=O(D), ~l=0.5$, because the distance between two 
point is at most $2D$, and the part of each edge $(u'v')$ in the "$m$-intersection" 
gadget is $0.5$.
The result in Section \ref{sec:k-inter} implies that the distortion is
$O(mD\cdot q^2)$.
\end{proof}

We note that this embedding is oblivious to the structure of
the graph which is assumed to be worst case. Namely, if we knew e.g., that $G_n$ can be embedded in the
planar grid, so that no two non-adjacent edges intersect, then we
could do much better.

\ignore{
Here we improve the simple embedding of the solid metric (Theorem
\ref{thm:27})  that is
induced by an {\em unweighted} graph proving a
  better dependence on  $D$ (that is always at most $n$ in
  this case).

\begin{theorem}\label{thm:52}
    Let $G_n$ be a unit-weighted graph that is known in advance.  There exists an online embedder that can embed $q$ points  from
    $X_G$ into ${\mathbb R}^2$ with  distortion $O(mD^{1/2}\log n \cdot q^2 +
    nD)$.
  \end{theorem}

\begin{proof}[Proof of Theorem \ref{thm:52}]
We note that for a worst case for $D= \theta(n)$, the above improves the
possibly cubic dependence in $n$ of Theorem \ref{thm:27}, to
$O(m\sqrt{n} + n^2)$.
}

\vspace{0.4cm}
\noindent
{\bf Improving the dependence on diameter, }
Next we improve the simple embedding of Theorem
\ref{thm:27}, and prove Theorem \ref{thm:52}.
The idea is to use a decomposition into small diameter components by
deleting a sublinear amount of edges. As a result, the embedding of
points inside each components, using Theorem \ref{thm:27}, will
improve, and at the same time the $k$-intersection gadget for interconnecting the
components 
will be smaller.

We stress that we do not assume any special structure of
the given graph.

We  use here a version of expander / low diameter  
decomposition e.g., as in \cite{Th}.

For a graph $G$ and $C \subseteq V(G)$, let $\partial (C) = \{(u,v) \in
E(G)|~ u \in C, v \notin C \}, ~ \delta(C) = |\partial(C)|$ and $vol(C) = \sum_{v \in C} deg(v)$. Let
the {\em conductance} of $C,$  
$\Phi(C) = \frac{\delta(C)}{\min (vol(C), vol(V(G) \setminus C)~)}$, and
  $\Phi_G = \min_{C \subset V(G)} \Phi(C)$.

\begin{observation}[Observation 1.1 in \cite{Th}]\label{obs:58}
Let $\beta \in (0, 1)$ and $G = (V, E)$ a graph with $|E|=m$, then  
 there exists a partition $V_1, \dots , V_k$ of 
$V$  such that  for every $i$ $\Phi_{G[Vi]} \geq \beta$
and $\Sigma_i{\delta(V_i)} = O(\beta m\log n)$.
\end{observation}
\begin{observation}\label{obs:51}
If  $\Phi_G \geq \beta$ for some $\beta \in (0, 1)$, then 
$\D(G) = O((\log |V(G)|)/\beta)$.
\end{observation}

\begin{figure}
\begin{center}
\includegraphics[scale = 0.6]{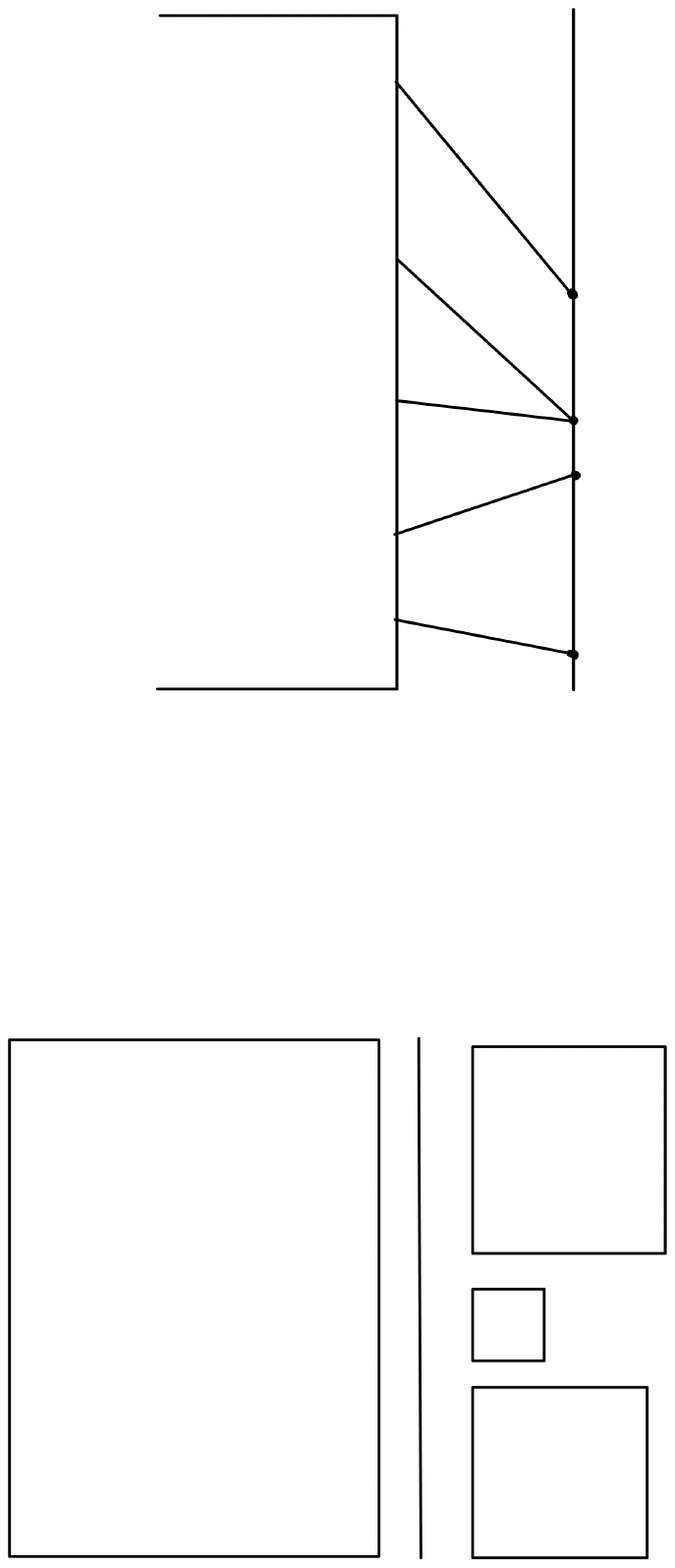}
\end{center}
\caption{The embedding using low diameter decomposition}
\label{fig:10}
\end{figure}

\begin{proof}[Proof of Theorem \ref{thm:52}]
We choose $\beta$ later and decompose $V(G)=V_1, V_2 \dots V_k$, as asserted by
Observation \ref{obs:58}, in which the diameter of $G[V_i]$ is bounded as asserted by
Observation \ref{obs:51}.  Let $E' = \cup_{i=1}^k \partial (V_i)$. 

We  enumerate the vertices in order of the partition: first all
the vertices of $V_1$, then $V_2$ and so on. 
Let $deg_i$ be the degree of $v_i$, and define $x_i = \Sigma_{j=0}^{i-1}{(deg_j+m/n)}$.
We will map each vertex to $\phi(v_i) = (x_i, 0)$, and 
define $S = m/q$.
We will map all the edges of $E'$ as in the naive solution - first spread 
them uniformly (without crossing) in the stripe between $[0,m]$ and $[S-1, m]$ 
using $0.25$-length of each edge. Then connect the edges endpoints using a
"$|E'|$-intersection" gadget with size $S$ (in the box $[0,S]\times[m,
m+S]$), see Figure \ref{fig:10}.

The expansion of spreading the lines is $O(m)$. In order to compute the contraction of
this phase we note that the minimal distance between the embedded lines is proportional 
to the minimum distance between their endpoint times a function of the angle
between the lines. More formally the minimal distance between $l_1$ and $l_2$ occurs at
points $a\in l_1, b\in l_2$ when $a$ or $b$ is at one of the endpoints of their
corresponding line. W.l.o.g assume $a$ is at one of $l_1$ endpoints and $e\in l_2$ one
of its endpoint, then $ab = ae \cdot sin(\angle bae)$. In our construction we get 
$\angle bae > \pi / 4$, meaning that the minimal distance between two lines is bounded
by the minimum of the distance in both endpoints (divided by $\sqrt{2}$).
Thus, the contraction is 
$O(D\cdot max(n/m, |E'|/S))=O(D\cdot max(n/m, \frac{\beta m\log n}{m/q}))=
O(D\cdot max(n/m,q\beta\log n))$. 
The parameters of the $|E'|$-intersection are $s=S, ~l=0.5, ~d=O(D),$ 
implying an expansion $O(Sq)=O(m)$ and the contraction 
$O(\frac{D\cdot |E'|\cdot q}{S})=O(\frac{D(\beta m\log n)q^2}{m}) = O(D\beta\log n
 \cdot q^2)$.

We are left to deal with edges inside each $V_i$. 
Define $E_i = E \cap (V_i\times V_i)$ as the set of edges inside the $i$'th 
component. Let $v_l$ and $v_r$ be the first and last vertices in $V_i$, and 
define $s_i = \Sigma_{j=l}^r{(deg_j+m/n)}$. Note that $x_{r+1} = x_l + s_i$, and 
$2|E_i|\leq s_i$. Define $S_i = s_i / q$, and map all the edges of $|E_i|$ as before, but on the negative side - spread them uniformly
(without crossing) in the section between $[x_l, -s_i]$ and $[x_{l}+S_i, -s_i]$ 
using $0.25$ of each edge. 
Then connect the edges endpoints using a "$|E_i|$-intersection" gadget of size $S_i$.

The parameters of the $|E_i|$-intersection are $s=S_i, ~l=0.5, ~d=O(D_i)$.
As before, the expansion along the paths is $O(s_i)$ and the contraction 
$O(D_i\cdot max(n/m, |E_i|/S_i)) \leq O(\beta^{-1}\log n\cdot max(1, \frac{s_i}{s_i/q}))
= O(\beta^{-1}\log n\cdot q)$ as a result of the bounded angle. 
Thus the expansion is $O(S_iq/0.5)=O(s_i)$, and the contraction is 
$O(\frac{D_i|E_i|q}{S_i})\leq O(\frac{(\beta^{-1}\log m)s_iq}{s_i/q}) = 
O(\beta^{-1}\log n\cdot q^2)$.
To sum up, the expansion is $O(max(m, s_i)) = O(m)$, and contraction
is \\
$O(max(D\cdot max(n/m, q\beta\log n), D\beta\log n \cdot q^2, 
\beta^{-1}\log n\cdot q, \beta^{-1}\log n\cdot q^2))
=$ \\ $O(max(D\cdot n/m, D\beta\log n \cdot q^2, \beta^{-1}\log n\cdot q^2))$.

Choosing  $\beta = \sqrt{D},$ the contraction 
is $\sqrt{D}\log n \cdot q^2 + D\cdot n/m$.
Overall the distortion is $O(m\cdot(\sqrt{D}\log n \cdot q^2 + D\cdot n/m)) = 
O(m\sqrt{D}\log n \cdot q^2 + nD)$. 
\end{proof}

\begin{proof}[Proof of Theorem \ref{thm:weighted}]
In the case of a weighted graph, we first normalize the weights such
 that the {\em minimum} weight is $1$. Define $D$ as the diameter of the graph after 
normalization. Then the exact same algorithm as in the proof of
Theorem \ref{thm:27} will yield a distortion of 
$O(mD\cdot q^2)$. Here $D$ becomes the aspect ratio of $(X_G,d)$ and
can be very large.  We design below 
 an algorithm which eliminates the dependence on aspect ratio.

We assume that a weighted graph $G=(V,E)$ with weights $d:E
\mapsto {\mathbb R}^+$ forms a metric space, namely where
$d(u,v)=distance(u,v)$ is given in advance.  Recall that a labeled
point-query will be
any $x \in (u,v)$ and $\alpha  \in [0,1]$, namely such that $d(x,u) =
\alpha d(u,v)$ and $d(x,v) = (1-\alpha)d(u,v)$.

\ignore{
First we present a reduction from weighted graphs to metric graphs given by the following lemma:

\begin{lemma}\label{lem:140}
    Given an online embedding algorithm of solid metric graphs with distortion $d(|V|, |E|, q)$,
    we present an online embedder of the metric of any solid weighted graph with distortion $d(|V|+|E|, 2|E|, q)$.
\end{lemma}
\begin{proof}
    Given a weighted graph $G = (V, E, w)$, we create a corresponding 
    metric graph $G = (V', E', w'), ~V\subseteq V', |V'| = |V|+|E|, |E'|=2|E|$ by
    subdividing all edges in the middle and creating new vertex there. Every edge 
    in $E'$ correspond to an half-edge in $G$, from $v\in V$ to $u\in V'/V$. The 
    solid metric spaces defined by those graphs are equivalent thus by embedding solid
    $G'$ we get an embedding of solid $G$ with distortion $d(|V|+|E|, 2|E|, q)$.
  \end{proof}
}
Our online embedding is based on an {\em offline} preprocessing phase
of $G=G_n$ that is known in advance. The grand goal of this
preprocessing is to cluster $G$ into small diameter (w.r.t that of
$G$) clusters. Then we allocate space in the plane for the online
embedding for each cluster, and a  disjoint space for the
interconnection between clusters. This clustering is going to be
recursive. More specifically, on each recursion level we decopose the
current cluster $H$ of diameter $D'$ (initially $G$ itself) into a
collection of connected subgraphs (subclusters) $\{C_i\}$, each of
diameter smaller than $D'/2$. Then we proceed recursively on each
subcluster.  While doing this clustering, an additional output is an
allocation of a box in the plane $B(C_i)$ for each $C_i$ inside the
box $B(H)$ in which $C_i$ is going to be embedded (if a point in it is
to be asked), and in addition, the allocation of a box in $B(H)$ in
which a $A$-intersection gadget will be implemented in order to embed
the (solid) edges connecting these clusters.  The formal algorithm is
presented below.

We denote by 
$D$ as the diameter of the graph $H$, and $D_i$ as the diameter of the
subclusters $C_i, i=1, \ldots$. 
We denote by $N$ as the size of the original graph (even in the recursive calls).

The input for a recursion step is: the weighted (sub)graph 
$H = (V, E, w), ~H \subseteq G$, $|V(H)|=n$ and a 
parameter $D'\geq D$ (upper bound on the diameter of $H$). In
addition, we are given a box $B(H)$ that is already allocated to $H$,
and that, by
induction, is  of width $n \cdot D'$ 
and height $2n\cdot D'$.

In the online phase every $C_i$ will be
embedded in its output box $B(C_i) \subset B(H)$, which for distinct
subclusters be pairwise
disjoint, all embedded alongside, and  near the bottom edge of $B(H)$, 
leaving a space for one $k$-intersection gadget at the top of $B(H)$,
see Figure \ref{fig:6}.

In addition, every interconnecting edge $(v,u)$ with
endpoints in distinct subclusters $C_i,C_j$ respectively, will be
relatively long with respect to $D'$ and it will be subdivided into 5
parts: 2 closer to $u$, symmetrically 2 closer to $v$ and  a middle part.   
The part of length $\frac{D'}{8N}$ closer to $u \in C_i$ will be
online embedded (when / if needed) in $B(C_i)$, starting  from the embedding of
$u$ ending at the top of $B(C_i)$. The next part will go up to an
$k$-interconnection gadget.  Symmetrically, this is done with the
two parts closer to $v$ with respect to
$C_j$. Finally, the middle part will be online embedded (if needed)
inside the $k$-interconnection gadget.

Thus, the offline preprocessing clustering needs to specify the
corresponding decomposition of $H$ into clusters $\{C_i\}$, the places
of the corresponding boxes $\{B(C_i)\}$ and the place of the
$k$-intersection gadget, see Figure \ref{fig:6}. 
This is done formally in what follows. 
\begin{figure}
\begin{center}
\includegraphics[scale = 0.6]{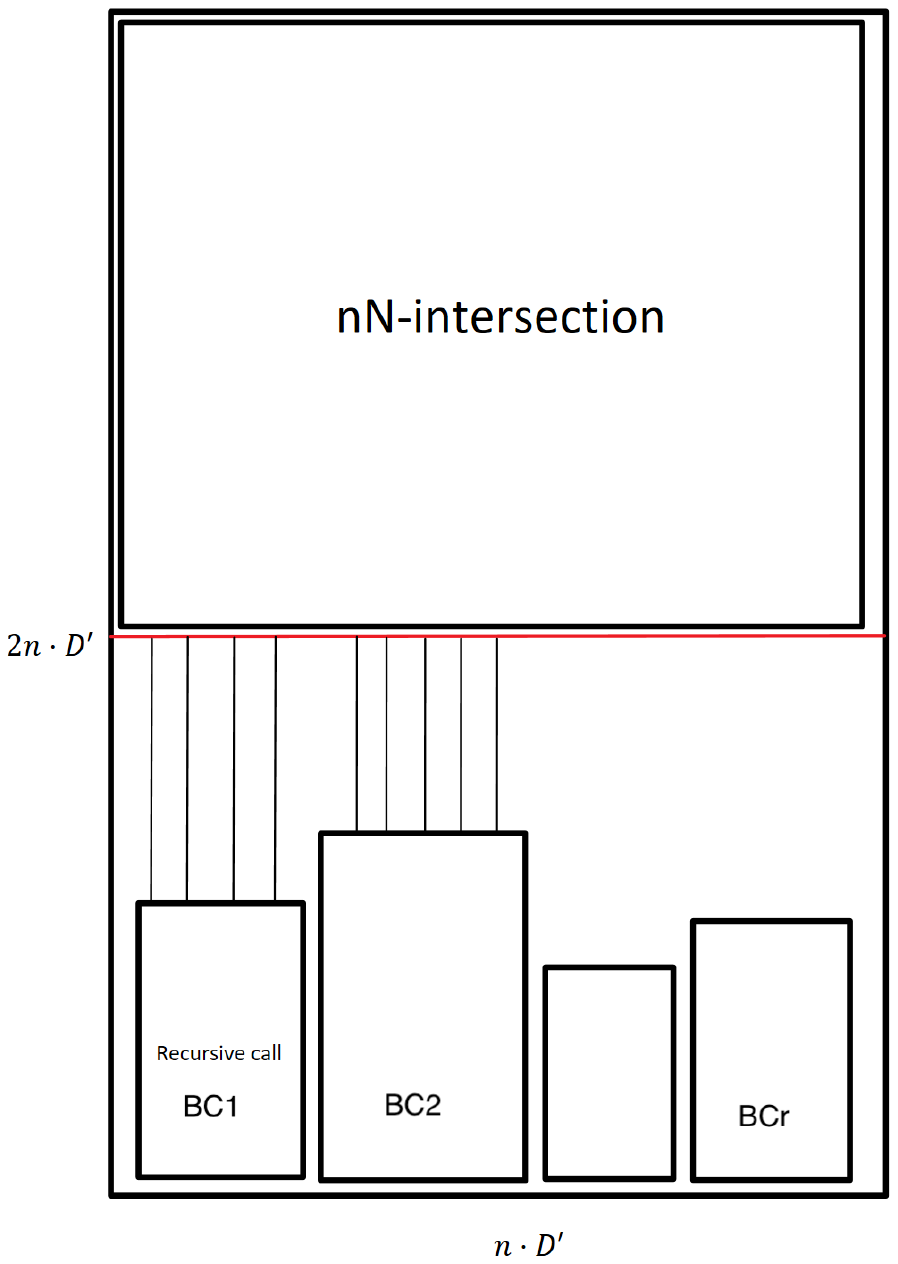}
\end{center}
\caption{Describing the intended embedding of the subclusters inside
  the box $B(H)$ at one recursion level.}
\label{fig:6}
\end{figure}

\vspace{0.4cm}
\noindent
{\bf The offline decomposition  algorithm:} (Recall $|V(H)|=n,~ |V(G)|=N$).

\begin{enumerate}
    \item Let $ E'=\{e\in E | w(e)<\frac{D'}{2n}\}$ and $H'=(V, E')$.
    \item Let $C_1, C_2 \dots C_r$ be the components of $H'$ and  $n_i = |C_i|$.
    Note that every edge $e\in E'$, namely that is inside a component has weight at most $\frac{D'}{2n}$.
   This implies that  $D_i < \frac{D'}{2}$, hence we set $D'_i = D'/2$.
    
    \item Embed recursively each $C_i$ to a box of width $n_i\cdot D'_i$ 
    and height $2n_i\cdot D'_i$,
    with the bottom side at $y=0$ with respect to $B(H)$,  and horizontal distance $\frac{D'}{2}$ between the 
    boxes. Note that $\sum_{i=1}^r{n_i\cdot D'_i+\frac{D'}{2}}\leq
    n\cdot D'$, and hence all the boxes embedding $C_1, \ldots ,C_r$
    will fit into a box of  width $nD'$.
    
    \item Each component $C_i$ is connected by  at most $n_i\cdot n$ edges to the graph
    $H/C_i$, and at most $n_i\cdot (N-n)$ edges to $G/H$.
    We use an "$nN$-intersection" gadget at the upper half of our
    box in order to map those edges. We allocate a $\frac{D'}{16N}$
    part  
    of each edge to bring it to the bottom of the gadget.
    
    For each edge leaving the graph $H$ we map its beginning of length
    $\frac{D'}{8N}$ inside our box. See Figure \ref{fig:6}.
      \end{enumerate}

\vspace{0.3cm}

The length of each edge between components in the gadget is its original length 
minus the parts which are already mapped in the recursive calls and to bring it to
the middle of the box, which is at least 
$\frac{D'}{2n} - 2(\frac{D'_i}{8N}+\frac{D'}{16N}) \geq
\frac{D'}{2n} - 2(\frac{\frac{D'}{2}}{4N}+\frac{D'}{16N})
\geq \frac{D'}{2n} - 2\cdot\frac{D'}{8n} = 
\frac{D'}{4n}$.
The length of each edge going up inside the gadget is the length of the prefix 
current call map minus the parts which are already mapped in the
recursive call and the length to bring it to the middle of the box, which
is $\frac{D'}{4N} - (\frac{D'_i}{8N}+\frac{D'}{16N}) \geq
\frac{D'}{4N} - (\frac{\frac{D'}{2}}{4N}+\frac{D'}{16N})
\geq \frac{D'}{8N}$.

The distance between any two point in the gadget is at most $O(D')$,
implying that  the
parameters of the "$nN$-intersection" are 
$s = nD', ~l = \frac{D'}{8N}, ~d = O(D')$. 
Thus we get that the expansion of the gadget in a specific call is 
$O(\frac{nD' \cdot 8N \cdot q}{D'})= O(nN q)$ and the contraction is 
$O(\frac{D' \cdot nN \cdot q}{nD'})=O(Nq)$.

Additionally the expansion resulted by  connecting the edges from their recursive
calls boxes to the bottom of the gadget is $\frac{nD'}{D'/8N} =
O(nN)$. In order to compute the contraction of the edges from the
$i$'th component,  we note that the distance between the edges in the plane is
$\frac{n_iD'_i}{n_iN} = \frac{D'}{2N}$, and the distance 
in the graph is at most the diameter in $C_i$ plus the allocated amount to bring them 
to the middle of the box, which is $O(D'_i)+2\cdot D'/16N = O(D')$. 
Thus the contraction is $\frac{O(D')}{D'/2N} = O(N)$.
Overall the distortion is $O(N^3\cdot q^2)$.
\end{proof}
  
\section{Proof of Theorem \ref{thm:UM-line} and \ref{thm:cluster}}\label{A:UM}

\begin{definition}[Shifted grid]
  \label{def:shifted-grid}
Let $G_{k,q} = (\{0,1,\ldots ,q^{1/k}-1\})^k$, with a fixed arbitrary
injection \\ $G:\{0, \ldots ,q-1\} \rightarrow G_{k,q}$, for which $G(0)=(0,
\ldots 0)$.

For a scale $s$ and a
 vector $v \in {\mathbb R}^k$ the shifted image $\phi_{v,s}: \nset \rightarrow
{\mathbb R}^k$ is defined $i \mapsto ~v+ s\cdot G(i)$.
\end{definition}

The idea behind the proof of  Theorem
  \ref{thm:cluster} is quite simple. The parameters are chosen
so that the metric induced on points between distinct subclusters of any class (cluster) $C$ is closed to the uniform metric. Hence we will embed each
cluster $C$, namely the at most $q$ points in it,  in a 'fat' point
using 
$\phi_{v,s}$ as the embedding,  where $v$ is the place we embed the
first point in $C$  and each point $x_t$ that is in a different
subcluster in $C,$ at
$\phi_{v,s}(t) = v
+ s\cdot G(t)$ for $s$ large enough and proportional to $\D(C)$.

This will expand the distances between points in $C$ that are of
distance large enough  by approximately $s\sqrt{k}q^{1/k}$ and with no contraction. Points in
subclusters of $C$ will be embedded 'inside' the corresponding 'fat'
points, as by the clustering properties, the diameter of a subcluster is
very small w.r.t $\D(C)$.
Further, by our
choice of parameters, there will be enough place for  each point
of every subcluster in the 'fat' point representing the cluster, with
minimal 'interaction' between subclusters.

To follow the strategy above, we need to show that the cluster-structure is well defined online, during
the partial exposure of points. This is the essence of 
Observation \ref{obs:73} below.

We assume that the points of $X$ are adversarially exposed in order
$x_0, x_1, \ldots ,x_{q-1}$ that is not known in advance. We assume
however, that $d$ is $(\gamma,\epsilon)$-recursively-clusterable and
that  $n, \gamma$ and $\epsilon$  are known in advance. We do not
assume that we know any other fact about the actual structure of the
recursive $(\gamma,\epsilon)$-partition of $X$. In particular, we do
not assume a prior knowledge on the number of classes at any recursion
level in the recursive decomposition of $X$.

\begin{notations}
  Let $(X,d)$ be $(\gamma, \epsilon)$-recursively clusterable, and
  $X=\{x_0, \ldots x_{t-1}\}$ enumerated according to the exposure order. 
The superscript $(t)$ is used to
  define properties / quantities of the exposed metric and of the
  embedding at time  $t$. Thus e.g., for the exposed sequence $x_0, x_1,\ldots ...$ we denote $X^{(t)} = \{x_0,
\ldots ,x_t\}$.

For
  $x,y \in X$ we denote $C(x,y)$ the smallest cluster  that contains
  both $x$ and $y$.  We denote $C^{(t)}(x,y) = C(x,y) \cap X^{(t)}$,
  namely the cluster containing $x,y$ at time $t$. We do
  not necessarily know $C(x,y)$ at time $t < q$, but we will show
  that for $x=x_j, y=x_\ell ~, j,\ell
  \leq t$,  namely, that are exposed at $t$ or before, we do know
  $C^{(t)}(x,y)$. 
\end{notations}

\begin{observation}
  \label{obs:73}
For $(X,d)$ as above, let $x=x_t$ and $y \in X^{(t-1)}$ for which $d(x,y)$ is the
smallest. Let $C^{(t)}(x,y)$ be the smallest cluster that contains both
$x,y$ at time $t$, and let $C_{\gamma}(x,y) = \{z \in X^{(t)}|~ d(z,x)
\leq \gamma d(x,y)\}$. Then  $C^{(t)}(x,y) =
C_{\gamma}(x,y)$. 
\end{observation}

\noindent

  \begin{proof}[Proof of Observation \ref{obs:73}]
 Assume first that $z \in C_{\gamma}(x,y)$, and assume for the
  contrary that $z \notin C^{(t)}(x,y)$. Let $C(z,y)=C^{(t)}(z,y)$ the
  smallest cluster that contains both $z,y$. Clearly $x \in
  C^{(t)}(z,y)$.
  The assumption that $z
  \notin C^{(t)}(x,y)$ implies that $z$ and $y,x$ are in distinct subclusters
  of $C^{(t)}(z,y)$. Since $C^{(t)}(z,y)$ is
  $(\gamma,\epsilon)$-clusterable,  it follows that $\D(C^{(t)}(z,y)) \leq \gamma \cdot
  d(z,x)$ while $\D(C^{(t)}(x,y)) \leq \epsilon \cdot \D(C(z,y))$.  We
  conclude 
  that $d(z,x) \geq \frac{\D(C(z,y))}{\gamma} \geq
  \frac{d(x,y)}{\epsilon\gamma}$ in contradiction to the fact that $z
  \in C_{\gamma}(x,y)$.

  For the other direction, let $z \in C^{(t)}(x,y)$. By
  definition $C^{(t)}(x,y)$ is also $(\gamma,\epsilon)$-clusterable,
  hence in the next level $y$ is separated from $x$ on account of
  $C^{(t)}(x,y)$ being the smallest containing $x,y$.
  Then, by clusterability of $C(x,y)$,  $\D(C(x,y)) \leq \gamma
  d(x,y)$. Hence,  $d(x,z) \leq \D(C^{(t)}(x,y)) \leq \gamma d(x,y)$ implying that $z \in
  C_{\gamma}(x,y)$.
\end{proof}

\begin{proof}[Proof of Theorem \ref{thm:cluster} ]
  We define next the online embedding $\phi: X \longrightarrow {\mathbb R}^k$, and denote
  the induced metric by $\tilde{d}$.

  Let $N = 2\gamma$.

  \begin{enumerate}
  \item $\phi(x_0)=\bar{0} = (0,\ldots ,0)$.
  \item At time $t$ embed $x=x_t$ as follows.

    Let $y \in X^{(t-1)}$ be the closest to $x_t$, and
    $C^{(t)}(x,y)=C^{(t)}_{\gamma}(x,y) = \{z \in X^{(t)}, ~
    d(z,x) \leq \gamma d(y,x)\}$. 

    Let $y_1,y_2 \in C^{(t)}(x,y)$ be the first and second exposed points in
    $C^{(t)}(y,x)$, for which $d(y_1,y_2) \geq d(y,x)$.  Note that $C^{(t)}(x,y)$ contains $x,y$ and hence $y_1
    \neq y_2$ are well defined. We set $L=L(C(x,y)) = d(y_1,y_2)$ and set
    $\phi(x)
    = \phi(y_1) + N L \cdot G(t)$.  
  \end{enumerate}
 The theorem now follows from Lemma \ref{lem:cluster}, items 3,
 and item 2, applied for $C^{(q)} = X$.
\end{proof}
\begin{lemma} 
  \label{lem:cluster}
\begin{enumerate}
\item For any cluster $C, ~ \tilde{\D}(C) \leq N\sqrt{k}q^{1/k}\D(C).$
 \item $\tilde{d} \leq 2\gamma^2 \sqrt{k}q^{1/k} \cdot d$.

  \item $\tilde{d}$ is not contracting.
  \end{enumerate}
\end{lemma}

\begin{proof}
For the first item consider a class $C \subseteq X$ with its first two 
exposed points $y_1,y_2$, and $L = d(y_1,y_2)$. Since all points in
$C$ are embedded by a shifted
embedding at $\phi(y_1)+ v + N\cdot L \cdot G(t)$ for positive vectors $v$ and
corresponding $t$, it immediately follows that
$\tilde{\D}(C) \leq N\cdot L \cdot \sqrt{k}(q^{1/k}-1) + \tilde{d}(C')$ where $C'$
is the subcluster in $C$ achieving the maximum embedded diameter.
The first item above corresponds
to the furthest distance between two points in the shifted $G_{k,q}$ with scale
$N\cdot L$.

The same argument can be made for $C'$ and corresponding scale $L'\leq
\D(C') \leq \epsilon L$. Thus plugging this for $C'$ in the above, and using induction we get.

$$\tilde{\D}(C) \leq NL\sqrt{k}(q^{1/k}-1) + N\cdot \epsilon L \cdot \D(C'') \leq
\ldots \leq $$ $$ NL\sqrt{k}(q^{1/k}-1) \sum_{i \geq 0} \epsilon^i \leq
NL\sqrt{k}(q^{1/k}-1) \frac{1}{1-\epsilon} \leq NL\sqrt{k}q^{1/k}
\leq N\sqrt{k}q^{1/k}\D(C)$$
where the last inequality is since $d(y_1,y_2) \leq \D(C)$.

For the second item let $x,y \in X$ be any two points and $C(x,y)$ the
smallest class containing both. Then by definition of clusterability
$ \D(C(x,y)) \leq \gamma d(x,y)$. Plugging the lower bound on
$\tilde{\D}(C(x,y))$ from 1. we get $\tilde{d}(x,y) \leq N \gamma d(x,y)
\sqrt{k}q^{1/k} \leq 2\gamma^2 \sqrt{k}q^{1/k} d(x,y),$ where the last
inequality is by our choice of $N$.

For the 3rd item let  $x=x_t$, and let 
    $y,C^{(t)}(x,y)=C_{\gamma}(x,y), y_1,y_2 \in C^{(t)}(x,y)$ and 
    $L =d(y_1,y_2)$ as defined in the embedding above.
    
  Consider first any $z =z_j \in C^{(t)}(x,y)$. Let $C(z) \subset
  C^{(t)}(x,y)$ the cluster of $z$ in the $(\gamma,
  \epsilon)$-clustering of $C^{(t)}(x,y)$. Since, by Observation
  \ref{obs:73}, $x$ form its own cluster in that clustering, $x \notin
  C(z)$. Let $z_0=x_r$ be the first exposed in $C(z)$. Then, by the
  definition of the embedding, $z_0$ is embedded in the shifted grid
  $G_{\phi(y_1),L}$ namely $\phi(z_0) = \phi(y_1) + NL \cdot
  G(r)$. Hence it follows that $\tilde{d}(z_0,x) \geq NL.$
  Further, using item 1. in the lemma, which we already proved,
  $\tilde{\D}(C(z)) \leq N\sqrt{k} q^{1/k}\D(C(z))$. Recalling that
  $\D(C(z)) \leq \epsilon \cdot \D(C^{(t)}(x,y))$ and that $d(z,x) \geq \D(C^{(t)}(x,y))/\gamma$,
  we conclude, 

  $$\tilde{d}(x,z) \geq \tilde{d}(z_0,x) - \tilde{d}(C(z)) \geq
  N\cdot L - \epsilon N\sqrt{k}
  q^{1/k}\cdot \D(C^{(t)}(x,y)) \geq $$
  $$
  N(~\D(C^{(t)}(x,y))/\gamma - \epsilon \sqrt{k} q^{1/k}\cdot
  \D(C^{(t)}(x,y)) ~ )
  \geq   Nd(x,z))(\frac{1}{\gamma}  - \epsilon \sqrt{k} q^{1/k}) $$
  where the second to last inequality is by the
  $(\gamma,\epsilon)$-clusterability of $C^{(t)}(x,y) > d(x,z)$ and
  the last is by our choice of parameters.  This proves that
  $\tilde{d}$ is not contracting on $C^{(t)}(x,y)$.
  
  For $z \notin C^{t)}(x,y),$ $z$ got separated from $y$ before or at time
  $t$. Let $C^{(t)}(z)$ be the largest cluster that contains $z$ and
  not $y$, and let $C^{(t-1)}(y,z)$ the smallest cluster that contains
  both. Let $y_1(z),y_2(z)$ be the first exposed points in $C^{(t-1)}(y,z)$. Then 

    $$\tilde{d}(z,x) \geq \tilde{d}(y_1(z),y_2(z)) -
    \tilde{d}(C^{(t)}(z)) - \tilde{d}(C^{(t)}(x,y)) \geq $$
    $$
    Nd(y_1(z),y_2(z)) - N\sqrt{k}q^{1/k}\D(C(z)) -
    N\sqrt{k}q^{1/k}\D(C^{(t)}(x,y)).$$ where the last inequality is
    by the definition of embedding and the first item in the Lemma.

    Recalling that $d(y_1(z),y_2(z)) \geq \frac{1}{\gamma}\D(C(y,z))
    \geq \frac{1}{\gamma} d(z,x)$
    and that $\D(C(z)), \D(C^{(t)}(x,y)) \leq \epsilon \cdot \D(C(y,z))$,
    and plugging into the above we get,

    $$\tilde{d}(z,x) \geq  d(z,x) \cdot N(~ \frac{1}{\gamma}  - 2\epsilon
    \sqrt{k}q^{1/k}~)  \geq d(z,x)$$ by our choice of parameters.
\end{proof}

      \begin{remark}\label{rem:12}
        In the above, we embedded the near uniform
        metric in ${\mathbb R}^k$ by using an arbitrary non-contracting mapping into the
        (shifted) grid $[q^{1/k}]^{k}$. This has the property that it
          can be done online, while expanding the uniform metric by at
          most 
          $\sqrt{k} \cdot q^{1/k}$. This is asymptotically optimal for
          $k = o(\log q)$. But for $k=\theta (\log q)$ a better
          embedding can be used:  One can use a fixed spherical code,
          e.g., as guaranteed by Johnson-Lindenstrauss \cite{JL} for
          embedding $K_n$ in $O(\log q)$-dim space with distortion $1+
          \epsilon$ (for arbitrary small fixed $\epsilon$).   Using
          such a  (shifted) embedding  instead of shifted grid will
          result in a distortion of $\gamma^2 q^{1/k}$ for $\epsilon <
          \frac{1}{2\gamma q^{1/k}}$ in Theorem
          \ref{thm:cluster}. However, this may be applied only for $k
          = \Omega (\log q)$ for which spherical codes with
          constant distortion exist.  For such construction, an
          online embedding of any clusterable metric as above into ${\mathbb R}^{O(\log q)},$ with
          $O(1)$-distortion, is achieved. See also 
          Remark \ref{rem:13}.
      \end{remark}

\section{Online embedding of Ultrametrics - proofs}\label{app:UM}

We use Theorem \ref{thm:cluster} to show that ultrametric can be
embedded into ${\mathbb R}^k$ with distortion $O(\sqrt{k}q^{2/k})$, and
in particular, with distortion $O(q^2)$ into the line and distortion
$O(q)$ in ${\mathbb R}^2$. We note that a UM $d$
is always $(1,\epsilon)$-clusterable for some $\epsilon$, but
$\epsilon$ depends on the metric and it is not necessarily made known
in the online process. We will first embed a UM $d$ into another UM
$d'$ to remedy the above obstacle.  

\begin{definition}
  Let $(X,d)$ be a UM, and  $\alpha >1$.
Let  $(X,d_{(\alpha)})$ be defined by
$d_{(\alpha)}(x,y) = \alpha^i$ for $i \in \Z$ the
smallest such that $d(x,y)
  \leq \alpha^i$. 
\end{definition}
\begin{observation}
  \label{obs:82}
  Let $(X,d)$ be a UM and $(X,d_{(\alpha)})$ as defined above. Then,
  \begin{enumerate}
  \item 
    $d_{(\alpha)}$ is a UM.
  \item  $d \leq d_{(\alpha)} \leq \alpha\cdot d$.
  \item  $d_{(\alpha)}$ can be computed online while $d$ is exposed.
    \item $d_{(\alpha)}$ is $(1,\frac{1}{\alpha})$-recursively clusterable.
\end{enumerate}
\end{observation}

\begin{proof}
  Consider the representation of $(X,d)$ by a
  weighted graph $G$, as asserted by (b) in the definition of UM
  (Section \ref{sec:notation}), and change (increase)
  the weight of each $e \in E(G)$ to the smallest $\alpha$-power, $\alpha^i,$ such
  that $w(e) \leq \alpha^i$.  This results in a weighted graph $G_{(\alpha)}$
  that by definition defines a UM that is exactly $d_{(\alpha)}$. Hence
  items (1), (2) in the claim directly follow. For (3) note that exposing
  $x$ with $d(x,y)$ for every previously exposed $y$ defines
  $d_{(\alpha)}(x,y)$.

  For the last item, it is enough to show that $d_{(\alpha)}$ is
  $(1,\frac{1}{\alpha})$-clusterable, as the metric induced on
  subclasses is also UM, and thus recursive clustering will
  automatically follow. Indeed, let $X = C_1 \cup \ldots
  C_\ell$ be the partition that is defined by the connected components
  (classes) 
  of $G_{\alpha}$ after deleting all edges of maximum weight
  $\alpha^a$. Then, then by definition, $d_{(\alpha)}(x,y) = \alpha^{a}$ for every
  $x,y$ belonging to different components proving that $\gamma=1$,
  and $d_{(\alpha)}(x,y)\leq \alpha^{a-1}$ for every component $C_k$ and
  $x,y \in C_k$, showing that $\epsilon \leq \frac{1}{\alpha}$.  
\end{proof}

  \begin{proof}[Proof of Theorem \ref{thm:UM-line}]
        For a UM $(X,d)$ with $|X|=q$, and fixed $k$, set
        $\alpha=2\sqrt{k} q^{1/k}$.  Then $d \leq
        d_{(\alpha)} \leq \alpha \cdot d$. Thus $d_{(\alpha)}$ distorts $d$ by at most
        $\alpha$. Further $d_{(\alpha)}$ is UM and by Observation \ref{obs:82}, it
        is $(1,1/\alpha)$-recursively-clusterable, meeting the
        condition of Theorem \ref{thm:cluster}. Hence Theorem
        \ref{thm:cluster} implies that $d_{(\alpha)}$ is embeddable into the
 ${\mathbb R}^k$ with distortion at most $O(\sqrt{k}q^{1/k})$, and thus $d$ with distortion
        at most $O(kq^{2/k})$.
      \end{proof}
      \begin{remark}
        \label{rem:13}
        For $k > \log q$, if a spherical code is used as discussed in
        Remark \ref{rem:12}, we set $\alpha = O(q^{1/k}) = O(1)$,
        achieving a total of $O(1)$ distortion. 
      \end{remark}

\section{The proof of Theorem \ref{thm:alpha-tree}}\label{sec:A-alpha}
We start with some comments and observations.

\begin{remark}\label{rem:32}
  For $x=x_t$ and $y \in X^{(t-1)}$ let $T^{(t)}(x,y)$ be largest
  subtree of $T$ that contains $x,y$. We note that $T^{(t)}(x,y)$ is
  rooted at some $v$ that is not necessarily exposed by time $t$, but
  that is well defined by $d(x,x_j), ~ j <t$. Hence we may assume that
  such $v$ (for the corresponding exposed $x$) is always exposed
  before $x$, by re-enumerating the vertices in the exposed sequence,
  resulting possibly in a sequence of $2q$ exposed vertices.
  Hence in all what follows we assumed that
  $v$ as above is always exposed before $x$, and it will be embedded
  too even if $v$ is not a part of the exposed sequence at
  all. Abusing notation we refer to this new order as
  $x_0, \ldots ,x_{q-1}$ (redefining the new number of exposed points
  that might $2q$ by $q$).  We note, however, that for $x,v$ as above
  it does not mean that $v$ is a father of $x$ in the original $T$,
  but rather that $v$ is a predecessor of $x$ (e.g., consider the case
  where $T$ is a path, then for $x,y$ as above $v$ is either $x$ or $y$).
 \end{remark}
 
  \begin{observation}
  \label{obs:93}
 \begin{enumerate}
   \item The subtree
      $T^{(t)} $ that is defined by $X^{(t)}$ is a also
      $\alpha'$-tree, for $\alpha' \leq \frac{\alpha}{1-\alpha}$ (as
      some paths in $T$ may becomes an edges in $T^{(t)}$).
      
 \item  For a subtree $T_v$ rooted at $v$ and with $u=  father(v)$ in $T$, $\D(T_v) \leq  \frac{2\alpha}{1-\alpha} d(u,v)$.
\end{enumerate}
\end{observation}

\begin{proof}[Proof of Theorem \ref{thm:natural_order} ]
  We assume that $T=(X,E)$ is an $\alpha$-tree and set $\beta =
  \frac{\alpha}{1-\alpha}$ as in the first item of Observation
  \ref{obs:93}. Namely, the partial tree $T^{(t)}$ is a $\beta$-tree
  (by item 1 in Observation \ref{obs:93}). We note that the natural
  order implies that for any $t$, $x_t$ is a leaf in $T^{(t)}$, and
  that $x_0=root$.
  
  We next define the online embedding $\phi: X \mapsto {\mathbb R}^k$, and denote
  the induced metric by $\tilde{d}$. The idea is similar to that in the proof
  of Theorem \ref{thm:cluster}. Namely, we embed each subtree $T_z$ in
  a shifted grid, allowing enough space to potentially embed $q-1$
  points that might be in $T_z$. 

  \vspace{0.3cm}
  \noindent
  {\bf The embedding}\\
  Recall that the grid $G_{k,q} = (\{0,1,\ldots ,n^{1/k}-1\})^k$ and $G:\nset \rightarrow G_{k,n}$, for which $G(0)=(0,
\ldots 0)$.  $G_{k,n}$ embeds naturally in $G_{k,2^k n}$ (namely in
the grid of sides of size $2n^{1/k}$) in the 'bottom
left' corner of it. In particular, every point $G(t), ~ t\leq n$ has
distance at least $n^{1/k}$ from any point $(a_0, \ldots
a_{n^{1/k}-1}) \in G_{k,2^k n}$ for which $\exists i, ~ a_i =
2(n^{1/k}-1)$. We will use in what follows the shifted embedding as discussed before,
but into this $G_{k,2^k q}$.  Let $N = \frac{16}{1-\beta}$

  \begin{enumerate}
  \item $\phi(root)=\bar{0} = (0,\ldots ,0)$.
  \item At time $t$ embed $x=x_t$ as follows.
    Let $y \in X^{(t-1)}$ be the closest to $x_t$ and let $i\in \Z$ the
    smallest such that $d(x,y) \leq q^{i/k}$. 
    We set
        $\phi(x) =\phi(y) + Nq^{i/k} \cdot G(t)$
      \end{enumerate}

      The theorem now follows by Lemma \ref{lem:natural}.  
      \end{proof}
      
      \begin{lemma}
        \label{lem:natural} Let $\tilde{d}^{(t)}$ be the metric that induced
        on $X^{(t)}$ in the embedding that is defined in Theorem
        \ref{thm:natural_order}. Let $\beta <
        \frac{1}{16\sqrt{k}q^{1/k}}$  Then
        \begin{enumerate}
        \item $\tilde{d} \leq  2N\sqrt{k}
        q^{1/k} d(x,y)$
          
          \item $\tilde{d} \geq d$.
                \end{enumerate}
      \end{lemma}
      \begin{proof}
        Assume that the above is true for $t' < t$. Let $x=x_t$, $y$ the closest to
        $x$ in $X^{(t-1)}$. Then, since $T^{(t)}$ is a $\beta$-tree
        (Observation \ref{obs:93} item 1), $y$ is the father of $x$ in
        $T^{(t)}$, and $x$ is a leaf in this tree.

        Let $z \in X^{(t-1)}$ and let $T^{(t)}(x,z)=T_v$ the smallest
        subtree containing both $x,z$. Consider first the case
        $z=y$. In this case $v=y$. Then by definition
        $\tilde{d}(x,y) \geq N d(x,y) > d(x,y)$. Let $i$ the
        corresponding integer in the definition of $\phi(x)$. Then,
        $\tilde{d}(x,y) \leq N \cdot 2\sqrt{k}\cdot q^{i/k} \leq
        2N\sqrt{k} q^{1/k} d(x,y)$. This proves both items in the lemma for
        $z=y$.

        Assume now that $z \neq y$. If $z$ is on the path $P(z,x)$ from $x$ to
        the root, and let $P(z,x)= (z=z_0, z_1, \ldots
        z_\ell=y, z_{\ell+1}=x)$. Then by a similar argument as above $\tilde{d}(z,y)
        \leq \sum_1^ {\ell+1} \tilde{d}(z_{i-1}, z_i) \leq 2N \sqrt{k}
        q^{1/k}\sum_i d(z_{i-1},z_i) \leq  2N \sqrt{k} q^{1/k} d(z,x)$,
        proving the first item in the lemma for $z,x$.  For $z$ that
        is not an ancestor of $y$, and $v$ such that $T_v$ is the
        smallest subtree containing $y,z$, 
        $\tilde{d}(z,x) \leq \tilde{d}(v,x) + \tilde{d}(v,z)$
        which by the argument above is at most $ 2N \sqrt{k}
        q^{1/k}(d(v,x)+d(v,z)) \leq 2N \sqrt{k}
        q^{1/k}d(z,x)$, completing the entire proof of the first
        item.

        For second item we use (for the first time) that $T$ is a 
        $\beta$-tree. For $z=y$ the claim is already proved above. Consider first
        $z \neq y$, and such that $z$ is on the path $P(z,x)$ from $x$
        to the root, and let
        $P(z,x)= (z=z_0, z_1, \ldots z_\ell=y, z_{\ell+1}=x)$. Bye the
        definition of embedding and the assumption on the order, each
        $z_i$ is embedded in a positively shifted grid w.r.t
        $\phi(z_{i-1})$. Hence it is enough to lower bound
        $\tilde{d}(z,x)$ by $\tilde{d}(z,z_1)$. Indeed,
        $\tilde{d}(z,z_1) \geq  N  d(z,z_1)$ (as for the argument for
        $z=y$ above). On the other hand  $d(z,x) \leq
        \frac{1}{1-\beta} \cdot d(z,z_1)$. Using the last two
        inequalities we get $\tilde{d}(z,x) \geq \tilde{d}(z,z_1) \geq
        N(1-\beta) \cdot d(z,x) \geq d(z,x)$ as needed.

        Finally, for $z$ not on the path to $y$, let $T_{w}$ be the
        smallest subtree containing both $z,y$. Let $z_1,y_1$ the sons 
   of $w$ on the paths from $w$ to $z,y$
        respectively. Note that by the assumption on the exposure order, $w,y_1$
        are exposed before $y$ and $w,z_1$ is exposed before $z$, and
        $w$ is the closest to both $z_1,y_1$ at the time of their exposure.
          Then,
  \begin{equation}
    \label{eq:1}
\tilde{d}(z,x) \geq \tilde{d}(z_1,y_1) - \tilde{\D}(T_{z_1}) -
  \tilde{\D}(T_{y_1}).
  \end{equation}

  Now, by symmetry assume that $d(w,z_1) \leq d(w,y)$. Let $i \leq j$
  such 
  that, $q^{\frac{i-1}{k}} < d(w,z_1) \leq
   q^{\frac{i}{k}}$ and $q^{\frac{j-1}{k}} < d(w,y_1) \leq
   q^{\frac{j}{k}}$. By the definition of the
  embedding,         $\phi(z_1) =\phi(w) + Nq^{i/k} \cdot G(t_{z_1})$
  and         $\phi(y_1) =\phi(w) + Nq^{j/k} \cdot
  G(t_{y_1})$. Suppose first that $i=j$. Then, both
   $y_1,z_1$ are mapped to distinct points in $G_{k,2^k q}$. Hence, 
  \begin{equation}
    \label{eq:2}
    \tilde{d}(z_1,y_1) \geq N q^{j/k}.
  \end{equation}
Otherwise, if $i < j$ then $z_1$ is mapped in a subgrid internal to
the 'bottom-left' cell in the grid in which $y_1$ is mapped. In that
case, since we use an embedding into $G_{k,2^kq}$ (that is of sides $2q^{1/k}$), 
\begin{equation}
  \label{eq:3}
    \tilde{d}(z_1,y_1) \geq N q^{j/k}/2.
\end{equation}

  Observation  \ref{obs:93} (2nd item) implies, $\D(T_{z_1}), \D(T_{y_1}) \leq
  \frac{2\beta}{1-\beta} d(w,y_1) \leq \frac{2\beta}{1-\beta} q^{j/k}$.  But as we already proved item 1 in
  the lemma, it follows that $\tilde{\D}(T_{z_1}) \leq 2N\sqrt{k}
        q^{1/k} 
  \D(T_{z_1}) $ and similarly $\tilde{\D}(T_{y_1})
  \leq 2N\sqrt{k}
        q^{1/k} \cdot
  \D(T_{y_1}) $. Plugging Equation (\ref{eq:2}) or (\ref{eq:3}) and these last
  inequalities, into Equation \ref{eq:1} implies that
$$ \tilde{d}(z,x) \geq  N q^{j/k}/2 -   2N\sqrt{k}
        q^{1/k} \cdot (\D(T_{y_1}) + \D(T_{z_1}))
\geq  N q^{j/k} \cdot (\frac{1}{2} - 2\sqrt{k} \frac{2\beta q^{1/k}}{1-\beta}
) \geq$$
$$ N \frac{d(z,x)}{2(1-\beta)} \cdot  (\frac{1}{2} - 2\sqrt{k}
\frac{2\beta q^{1/k}}{1-\beta}) \geq d(z,x).$$ Where the last
inequality is by our assumption on $\beta$ and $N$. 
      \end{proof}

We now Prove Theorem \ref{thm:tree-to-natural},  showing how to online embed $\alpha$-tree into a $\alpha'$-tree for which
the exposure of the later is guaranteed to be the natural order (on
it), and $\alpha'$ is slightly larger than $\alpha$.  For simplicity
we  prove the theorem for $\alpha
 \leq \frac{1}{5}$. The same
 proof holds for any fixed $\alpha <1$, with a slightly different
 computation in the last step.

 \vspace{0.4cm}
 \noindent
\begin{proof}[Proof of Theorem \ref{thm:tree-to-natural}]
  The top level idea is the following. Suppose that $T^{(t-1)}$ is
  already embedded, and that $x=x_t$ is the next exposed point. As
  explained before, $x$ defines a unique extension of $T^{(t-1)}$
  denoted  $T^{(t)}$, either by
  subdividing a directed edge $(u,v)$ into $(u,x), (x,v)$ such that
  $d(u,v)=d(u,x) + d(x,v)$, or by inserting a new edge $(z,x)$ and so
  that $x$ is a leaf in $T^{(t)}$ (see Remark \ref{rem:32}).

  In the first case define $\tilde{T}$ by adding the directed edge  $(u,x)$ with length
 $\tilde{d}(u,x)= d(u,x)$. Note that $v$ becomes a father of $x$ in
 $\tilde{T}^{(t)}$ while it is a son of $x$ in $T^{(t)}$. In the 2nd
 case we just add 
 the edge $(z,x)$ with length $\tilde{d}(z,x) = d(u,x)$.

 \begin{remark}
 (a) If $\tilde{T}^{(t-1)}$ has natural order on $X^{(t-1)}$ then so does
 $\tilde{T}^{(t)}$ as we never subdivide an edge. (b) In the second case, this is just a trivial
augmentation of $\tilde{T}^{(t-1)}$ by $x$ preserving the natural order, and distance
of the newly added edge. The first case, distances can only
become larger. Hence, it directly follows that the resulting
$\tilde{d}$ is not contracting, and that $T^{(t)}$ is
an $\frac{\alpha}{1-\alpha}$-tree, as by Observation \ref{obs:93} (1st
item), $T^{(t-1)}$ is a $\frac{\alpha}{1-\alpha}$-tree and this is kept for every step.
\end{remark}

To bound expansion it is enough to bound the
edge-expansion for every original edge $e \in T$. Since expansion is
incurred only when predecessors  are exposed after their ancestors,
 it is enough to
consider a sequence of  exposures in a reverse order to a root-to-leaf
path.
Namely, prove a bound on the edge
expansion  where $T$ is a simple path. This is done next, by
induction on the length of the path.

We set $\beta = \frac{\alpha}{1-\alpha} < 1/4$, and recall that any subtree
of $T$ is a $\beta$-tree.

Assume that $T$ is the path $T=(0,1,2, \ldots m)$ and that the
exposure order is $(0, \pi(1), \ldots \pi(m))$ for
$\pi \in {\mathcal S}_m$. Assume that $\pi(1) = \ell$. Then the
resulting $\tilde{T}^{(m)}$ for the $m$-length path will be the rooted
tree at $0$, with a single out going edge $e=(0,\ell)$ of length
$d(0,\ell)$ and where $\ell$ has two subtree, $\tilde{T}_L$ that
contains all vertices in $[\ell-1]$ (as they are exposed after $\ell$
and will be inserted above $\ell$ in a subtree $\tilde{T}_L$ case 1
above), and a disjoint subtree $\tilde{T}_H$, that contains
$[m]\setminus [\ell]$ namely all vertices of indices larger than $\ell$
and that are exposed after it. Moreover, we note that 
the embedding of $\tilde{T}_H$ is a recursive embedding of the path
$(\ell+1, \ldots, n)$, where $\ell$ serves as a root (see Figure
\ref{fig:9}) .

\begin{figure}
\begin{center}
\includegraphics[scale = 0.7]{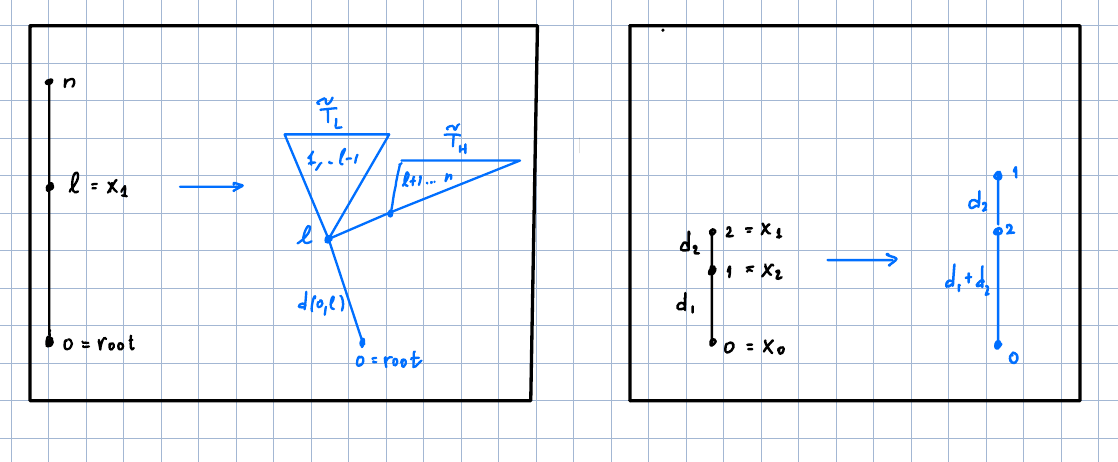}
\end{center}
\caption{The embedding of a (di)path: $o$ is exposed first then
  $\ell$. The left picture shows the embedding into $\tilde{T}$ where
  $\tilde{T}_L$ is the embedding of the points in $[\ell-1]$ and
  $\tilde{T}_H$ is the embedding of the points in $[m]-[\ell]$.
  The right box shows the base case: embedding of a 2-path of 'bad'
  order in natural order.
}
\label{fig:9}
\end{figure}

We now prove by induction on the structure, namely on subtrees that
the edge expansion is bounded by $c = 1+2\beta = \frac{1+\alpha}{1-\alpha}$.

The base case is for the path $(0,1,2)$, the only non-trivial exposure
order is $(0,2,1)$ which results the tree $\tilde{T} = (0,2,1)$
with edge length $\tilde{d}(0,2)= d(0,2), ~ \tilde{d}(2,1) = d(2,1)$,
see Figure \ref{fig:9}.
The only expanded original edge  is $e=(0,1)$ for which
$\tilde{d}(0,1) = d(0,2) + d(2,1) = 2d(1,2) + d(0,1) \leq 2\beta
d(0,1) + d(0,1)$, as $T$ is a $\beta$-tree.
We get that $\tilde{d}(0,1) \leq (2\beta +1) d(0,1)$ as needed.

Now for the general case, a special edge in $T_L$ is $(0,1)$. For it
we get $\tilde{d}(0,1) = d(0,\ell) + \tilde{d}_L(\ell,1)$ where
$\tilde{d}_L(\ell,1)$ is the appropriate distance in
$\tilde{T}_L$. Assuming inductively an expansion $c$ for any edge in $\tilde{T}_L$ we get,
$\tilde{d}_L(\ell,1) \leq c \cdot d(\ell,1) \leq c \cdot d(0,1) \cdot
\frac{\beta}{1-\beta}$ where the last inequality is by the
assumption that $T$ is an $\beta$-tree.
We conclude that  $\tilde{d}(0,1) \leq
d(0,1) \cdot (1+ \frac{c\beta }{1-\beta})$. Pluggin the induction
hypothesis for $c$ we get, $\tilde{d}(0,1) \leq (1+2\beta)d(0,1)$ for
$\beta < 1/4$.
 For any other edge in $(i,i+1), ~ i < \ell$ the expression is
similar.

Now for the edge $(\ell,\ell+1) \in T_H$, we note that the reasoning
is the same as for the edge $(0,1)$ where $\ell$ serves as the root of
$T_H$. Hence this ends the proof of the theorem.
\end{proof}
\vspace{0.4cm}
\noindent
\begin{proof}[Proof of Theorem \ref{thm:alpha-to-beta}]
  Let $T=(X,E)$ with root $0$ be an $\alpha$-tree. Let $\alpha' =
  \frac{\alpha}{1-\alpha}$, and let $\beta < \alpha $ be
  fixed. Assume that the order $(0,x_1, \ldots x_{n-1})$ is a natural
  order on $T$. Let $T^{(t)}$ the exposed tree at time $t$. Recall
  that $T^{(t)}$ is $\alpha'$-tree, and for
  every edge $e=(u,v) \in T^{(t)}$ let $w(e) = (\frac{1}{\beta})^i$ for
  $i\in \Z$ such
  that $(\frac{1}{\beta})^{i-1} < d(u,v) \leq (\frac{1}{\beta})^i$. We
  note that $w(e)$ is well defined by the time both $u,v$ are exposed
  by the assumption on the natural order. Further, let $d_w$ be the
  metric induced on $T^{(t)}$ by the weights $w: E(T) \longrightarrow
  {\mathbb R}$.  Then clearly
  \begin{equation}
    \label{eq:4}
    d \leq d_w \leq \frac{d}{\beta} 
  \end{equation}
  
We note however that $T_w$ is not a $\beta$-tree which we are going
to correct next.

Consider any node $v$ (and in particular $0$) in $T^{(t)}$.
Let $T_v(w(e))$  contain all edges of weight $w(e)$
in $T_v$. Note that $T_v(w(e))$ forms a directed subtree of $T_v$
rooted at $v$, and
any of its leaves forms a subtree in which all edges are of weight at
most $\beta \cdot w(e)$, by the assumption on $w$.  Let $V(w(e))$ the
set of vertices in $T_v(w(e))$. We embed each $x \in V(w(e))$ as a
direct son of $v$ with edge length $w(e)$. We do that from the root
$0$, following the exposure order, forming the tree $\tilde{T}$. Since order is natural by
assumption, this imposes a natural order on $\tilde{T}$ as well.  

We note that $\tilde{T}$ is indeed a $\beta$-tree by definition.
Further, note that the distance $\tilde{d} \leq d_w$ with possible
strict inequality as vertices become closer to the root, and this is
true for any subtree as well. A vertex $x$ in $T_v$ that
become a son of $v$, namely at distance $\tilde{d}(v,x) = w(e)$ has
original distance $d_w(v,x) = depth_{T_v}(x) \cdot w(e)$. But since originally
$T$ was an $\alpha'$-tree, and $depth_{T_v}(x)$ in $T_w$ does not change in
comparison to $T$, we conclude that
$(\alpha')^{depth_{T_v}(x)} \geq \beta$ and hence
$\tilde{d}(v,x) \leq \frac{\log 1/\beta}{\log 1/\alpha'} \cdot w(e)$.

Since the only contraction is inside each subtree as described, and
the 'global' tree structure remains the same, the contraction is bounded by the
contration in each subtree. 
We conclude that the total contraction is bounded by $\frac{\log
  1/\beta}{\log 1/\alpha'}$ and local distortion is bounded by $\frac{\log
  1/\beta}{\beta \cdot \log 1/\alpha'}$. Plugging in the expression
for $\alpha'$ implies the result.
  \end{proof}

\end{document}